\documentclass[aps,pra,twocolumn,notitlepage,nofootinbib,superscriptaddress,longbibliography]{revtex4-2}

\usepackage[utf8]{inputenc}
\usepackage[T1]{fontenc}
\usepackage{amsmath,amssymb,amsthm}
\usepackage{times}
\usepackage{physics}
\usepackage{bm}
\usepackage{hyperref}
\usepackage{enumerate}
\usepackage{enumitem}
\usepackage{color}
\usepackage{graphicx}
\usepackage{tikz}
\usepackage{orcidlink}
\usetikzlibrary{positioning,arrows.meta}

\newtheorem{axiom}{Axiom}
\newtheorem{definition}{Definition}
\newtheorem{theorem}{Theorem}
\newtheorem{corollary}{Corollary}
\newtheorem{proposition}{Proposition}
\newtheorem{lemma}{Lemma}
\newtheorem{remark}{Remark}

\begin{document}

\title{Information Theory of Action : Reconstructing Quantum Dynamics from Inference over Action Space}

\author{Fabricio Souza Luiz \orcidlink{0000-0002-6375-0939}}
\email{fsluiz@unicamp.br}
\affiliation{Instituto de F\'\i sica Gleb Wataghin, Universidade Estadual de Campinas, Campinas, SP, Brazil}

\author{Marcos César de Oliveira \orcidlink{0000-0003-2251-2632}}
\email{marcos@ifi.unicamp.br}
\affiliation{Instituto de F\'\i sica Gleb Wataghin, Universidade Estadual de Campinas, Campinas, SP, Brazil}

\date{\today}

% ---------------------------------------------------
% ABSTRACT
% ---------------------------------------------------
%\begin{abstract}
%We derive the quantum formalism from information-theoretic principles applied to action space, without assuming any quantum postulates. The central object is the density of action states $g(A;b|a)$, counting the number of ways a system evolves between configurations with total action $A$. Maximum entropy inference on this space yields a probability distribution whose Fisher information establishes a minimal resolvable action scale $\Delta A_{\min} = \hbar$. We prove that when action contributions are indistinguishable at this scale, consistency requirements force their combination via complex amplitudes with phase $e^{iA/\hbar}$, the quantum phase emerges as a logical necessity, not a postulate. A key result is universal infinitesimal indistinguishability: for any pair of trajectories, $|\delta S|/\Delta S_{\min} \propto \sqrt{dt} \to 0$ as $dt \to 0$, making coherent summation mandatory at every time step. From the resulting propagator $K(b|a) = \int g(A;b|a)e^{iA/\hbar}dA$, we derive the Lagrangian via Galilean invariance, the Hilbert space structure via Stone's theorem, and the Schr\"odinger equation from the short-time limit. The parameter $\hbar$ is determined empirically, not theoretically. This first paper establishes the foundations; a companion paper addresses measurement, the Born rule, and entanglement.
%\end{abstract}

\begin{abstract}
We develop an information-theoretic reconstruction of quantum dynamics based on inference over action space. The fundamental object is a density of action states encoding the multiplicity of dynamical alternatives between configurations. Maximum-entropy inference introduces a finite resolution scale in action, implying that sufficiently close action contributions are operationally indistinguishable. We show that this indistinguishability, together with probability normalization and action additivity, selects complex amplitudes and unitary evolution as the minimal continuous representation compatible with action additivity, probability normalization, and inference under finite resolution. Quantum interference and unitarity therefore emerge as consequences of these assumptions rather than independent postulates.
From the resulting propagator, the Lagrangian, Hilbert-space structure, and Schrödinger equation follow as derived consequences. In the infinitesimal-time limit, action differences universally fall below the resolution scale, making coherent summation the minimal consistent description at every step. The numerical value of the action scale is fixed empirically and identified with $\hbar$.
\end{abstract}

\maketitle

% ---------------------------------------------------
\section{Introduction}
\label{sec:intro}

Nearly a century after its formulation, the conceptual foundations of quantum mechanics continue to be the subject of active investigation and discussion~\cite{vonNeumann1932,Dirac1930}. Despite its extraordinary empirical success, the standard Hilbert-space formalism brings together physical postulates and epistemic rules whose logical interrelations are not always made explicit--the linear structure of states, the complex nature of the Hilbert space, the Born rule, and the measurement postulate are typically introduced axiomatically and largely independently~\cite{Weinberg2015}. This independence obscures which elements of the formalism are logically fundamental and which are consequences of deeper consistency requirements. As a result, the formal structure of the theory is mathematically precise and powerful, while its conceptual organization admits multiple interpretations.

At the same time, action plays a central role in all modern formulations of classical physics. Hamilton's principle provides a unified description of mechanics, electromagnetism, general relativity, and field theory. Feynman's path integral identifies the phase $e^{iS/\hbar}$, where $S$ is the path action, as the cornerstone of quantum interference~\cite{Feynman1948,Schulman1981,Kleinert2009}, yet the integral itself rests on a formal measure that lacks rigorous definition and treats trajectories as primitive microscopic objects, even though they play no observable role in quantum theory. The gap between the operational meaning of quantum amplitudes and the formal structure of the path integral highlights the need for a more fundamental perspective on the role of the action in quantum theory.

Information theory has emerged as a powerful framework for deriving physical laws from principles of inference rather than dynamical postulates. Jaynes demonstrated that the maximum entropy (MaxEnt) principle provides a uniquely consistent method for assigning probabilities given incomplete information~\cite{Jaynes1957a,Jaynes1957b}, and the Shore–Johnson theorem established MaxEnt as the only inference rule compatible with rational updating~\cite{ShoreJohnson1980}. Cox showed that probability theory itself can be derived from consistency requirements on plausibility assignments~\cite{Cox1946,Cox1961}. More recently, entropic and information-theoretic reconstructions of physical theories have shown that probabilistic structure can emerge from constraints and symmetries rather than axioms~\cite{Hardy2001,Chiribella2011,MasanesMuller2011,Caticha2011,Goyal2010,Reginatto1998,HallReginatto2002,Frieden1998}. In this sense, probabilistic structure is not freely postulated, but constrained by consistency, symmetry, and compositional requirements.

These developments motivate a deeper question:
\emph{Can quantum mechanics be derived from a theory whose only inputs are (i) action, (ii) the multiplicity of action values, and (iii) inference based on incomplete information?}

In this work we answer this question affirmatively by developing what we call the \emph{Information Theory of Action} (ITA).
Our goal is not to postulate quantum amplitudes, Hilbert spaces, or operators, but to show that these structures arise necessarily once one combines (i) action additivity, (ii) multiplicity of dynamical alternatives, and (iii) rational inference under finite resolution. Figure~\ref{fig:logical_structure} summarizes the Logical structure of the theory.  The present approach differs from previous reconstruction programs in its choice of primitive object. Existing formulations typically
take probabilities, operational outcomes, or state spaces as primary and derive dynamical structure from consistency requirements on their composition. In contrast, Information Theory of Action takes the multiplicity of dynamical alternatives labeled by additive action as fundamental, with probabilistic structure arising only after inference is performed relative to this multiplicity. The reconstruction therefore proceeds from dynamical structure to probabilistic representation, rather than the reverse, providing a complementary perspective to operational and generalized probabilistic approaches.

%----------------------------------------------------------------------------------------------
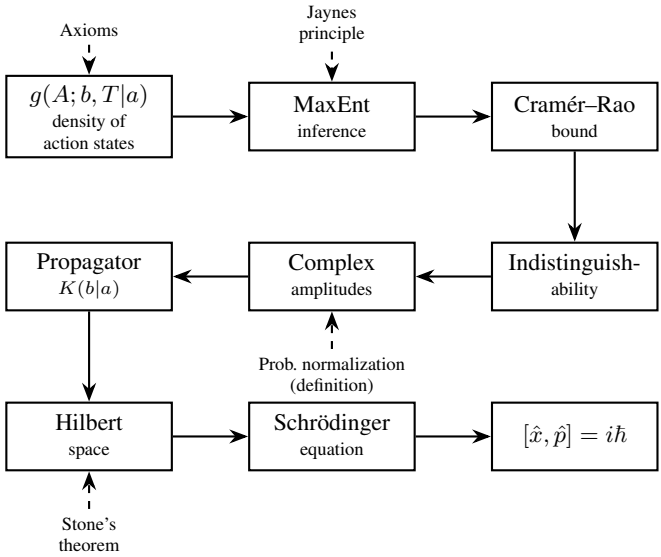
\begin{figure}[t]
\centering
\begin{tikzpicture}[
    node distance=1.0cm,
    box/.style={rectangle, draw=black, thick, minimum width=2.2cm, minimum height=0.9cm, align=center, font=\small},
    arrow/.style={-{Stealth[length=2.5mm]}, thick},
    label/.style={font=\scriptsize, align=center}
]
% Main derivation chain
\node[box] (g) {$g(A;b, T|a)$\\[-1pt]\scriptsize density of\\[-2pt]\scriptsize action states};
\node[box, right=of g] (maxent) {MaxEnt\\[-1pt]\scriptsize inference};
\node[box, right=of maxent] (CR) {Cram\'er--Rao\\[-1pt]\scriptsize bound};
\node[box, below=1.2cm of CR] (indist) {Indistinguish-\\[-1pt]\scriptsize ability};
\node[box, left=of indist] (complex) {Complex\\[-1pt]\scriptsize amplitudes};
\node[box, left=of complex] (K) {Propagator\\[-1pt]\scriptsize $K(b|a)$};
\node[box, below=1.2cm of K] (hilbert) {Hilbert\\[-1pt]\scriptsize space};
\node[box, right=of hilbert] (schrod) {Schr\"odinger\\[-1pt]\scriptsize equation};
\node[box, right=of schrod] (CR2) {$[\hat{x},\hat{p}]=i\hbar$};

% Arrows
\draw[arrow] (g) -- (maxent);
\draw[arrow] (maxent) -- (CR);
\draw[arrow] (CR) -- (indist);
\draw[arrow] (indist) -- (complex);
\draw[arrow] (complex) -- (K);
\draw[arrow] (K) -- (hilbert);
\draw[arrow] (hilbert) -- (schrod);
\draw[arrow] (schrod) -- (CR2);

% Side inputs
\node[above=0.4cm of g, font=\scriptsize, align=center] (axiom) {Axioms};
\draw[arrow, dashed] (axiom) -- (g);

\node[above=0.4cm of maxent, font=\scriptsize, align=center] (jaynes) {Jaynes\\principle};
\draw[arrow, dashed] (jaynes) -- (maxent);

\node[below=0.5cm of complex, font=\scriptsize, align=center] (unit) {Prob.\ normalization\\(definition)};
\draw[arrow, dashed] (unit) -- (complex);

\node[below=0.5cm of hilbert, font=\scriptsize, align=center] (stone) {Stone's\\theorem};
\draw[arrow, dashed] (stone) -- (hilbert);

\end{tikzpicture}
\caption{Logical structure of Information Theory of Action. The derivation flows from the axiomatically defined density of action states $g(A;b, T|a)$ through maximum entropy inference, the Cram\'er--Rao bound establishing an action resolution scale, and indistinguishability combined with probability normalization forcing complex amplitudes and unitarity. The propagator $K(b|a)$ then yields Hilbert space structure via Stone's theorem, culminating in the Schr\"odinger equation and canonical commutation relations. Dashed arrows indicate external inputs: the axioms, Jaynes' inference principle, probability normalization (a definitional requirement, not a physical assumption), and Stone's theorem from functional analysis. Interference and unitarity emerge as consequences, not inputs.}
\label{fig:logical_structure}
\end{figure}

If dynamical alternatives are labeled by additive action $A$, and action values cannot be resolved below a finite scale, then consistent probabilistic evolution selects complex amplitudes of the form $e^{iA/\hbar}$ as the minimal continuous representation compatible with the stated assumptions. Quantum mechanics therefore appears as the minimal continuous representation under assumptions compatible with action additivity, inference under finite resolution, and probability normalization. The basic physical object is not a trajectory, a wavefunction, or an operator, but a density of action states $g(A;b, T|a)$, which counts how many ways the system may evolve from an initial configuration $a$ to a final configuration $b$ with total action $A$. This density plays the same role for action that the density of states $\Omega(E)$ plays for energy in statistical mechanics.

Given $g(A;b,T|a)$, the state of knowledge about the dynamical process is described by a probability density over the joint space of actions and endpoints, $P(A,b|a)$, obtained by maximizing the relative entropy subject to constraints. The resulting exponential family $P\propto g(A;b,T|a)e^{-\eta A}$ carries a natural information metric. From the Fisher information we derive a Cramér–Rao bound~\cite{Cramer1946,Rao1945,amari2000methods} establishing a minimal resolvable action difference $\Delta A_{\min}=1/\eta$. This operational indistinguishability defines when two contributions to a dynamical process cannot be told apart.

Indistinguishability forces coherent combination of action contributions. Following the approach of Goyal, Knuth, and Skilling~\cite{Goyal2010}, we show that consistency requirements on how indistinguishable alternatives combine, analogous to Cox's axioms for probability, uniquely determine that they must be represented by complex amplitudes. Because action is additive, the amplitude must satisfy a multiplicative functional equation. We prove that the only continuous solution compatible with physical requirements is the complex phase $e^{i\eta A}$. Thus, the phase of quantum mechanics emerges not as a postulate, but as the unique consistent representation of contributions that cannot be distinguished at the action scale $1/\eta$.

% This leads directly to a derived propagator,
% \[
% K(b|a)=\int dA\,g(A;b|a)\,e^{i\eta A}.
% \]
%No Hilbert space, no paths, no operators, and no wavefunction have been assumed--the propagator is a consequence of action multiplicity and indistinguishability.

Experimental interference phenomena identify the empirical
action scale with $\hbar$. By comparing the action resolution $1/\eta$ implied by the MaxEnt distribution with the observed scale of quantum interference, we identify $\eta=1/\hbar$. With this identification, the propagator reproduces the correct quantum phases. If $g(A;b,T|a)$ admits a trajectory representation, substituting it into the expression for propagator, $K$, recovers the Feynman path integral. We derive the short-time behavior of $g(A;b,T|a)$ from locality requirements and the composition law, showing that it must be sharply peaked around the classical action with a specific Gaussian profile.

Once the propagator is known, we show that it defines a linear, norm-preserving semigroup on $L^2(\mathcal{Q})$. We provide a complete proof of unitarity based on time-reversal symmetry. By Stone's theorem, this implies the existence of a self-adjoint Hamiltonian $\hat H$ and unitary evolution $U(t)=e^{-i\hat Ht/\hbar}$. %Thus the entire Hilbert-space structure of quantum mechanics emerges from the action-based formalism.
The short-time expansion of the propagator yields the Schrödinger equation $i\hbar\partial_t\psi = \hat{H}\psi$ and the canonical commutator $[\hat{x},\hat{p}] = i\hbar$. The entire Hilbert space structure emerges from Stone's theorem~\cite{Stone1932,Stone1932a} applied to the unitary semigroup generated by the propagator. This completes the foundational reconstruction.

The present paper is restricted to the foundations of quantum dynamics. Measurement, entanglement, and the operational interpretation of the Born rule are intentionally deferred to a following work. No claim is made here to resolve the measurement problem or to derive collapse dynamics. Rather, we focus on establishing why coherent unitary evolution, complex amplitudes, and the Schr\"odinger equation are unavoidable consequences of action indistinguishability and probability normalization.

The present work should be understood as a reconstruction
of quantum dynamics rather than a claim of uniqueness among
all conceivable physical theories. The goal is to show that,
once action additivity, finite inferential resolution, and
probability normalization are imposed, the standard quantum
formalism arises as a minimal consistent dynamical
representation.  Alternative formulations based on different
assumptions are not excluded. The present one offers a path, not  to reinterpret quantum mechanics,
but to clarify why its mathematical structure is naturally
selected by consistency requirements applied to inference
over additive action.

The paper is organized as follows. Section~\ref{sec:axioms} introduces the physical axioms: configuration space, action space, and the density of action states $g(A;b,T|a)$. Section~\ref{sec:maxent} develops MaxEnt inference on the joint space $(A,b)$. Section~\ref{sec:CR} establishes the Cramér--Rao bound and action resolution scale. Section~\ref{sec:amplitudes} derives complex amplitudes from indistinguishability via a Cox-type argument. Section~\ref{sec:propagator} constructs the propagator, and Section~\ref{sec:hbar} identifies $\eta = 1/\hbar$ through scale unification. Section~\ref{sec:shorttime_g} derives the short-time structure of $g$ from locality and composition. Section~\ref{sec:lagrangian_emergence} shows the emergence of the Lagrangian and variational principle. Section~\ref{sec:hilbert} establishes the Hilbert space structure, while Section~\ref{sec:dynamics} derives the Schrödinger equation and canonical commutation relations. Finally, Section~\ref{sec:conclusion} encloses the paper.

Each step follows from the axioms and consistency requirements, with minimal additional assumptions clearly identified.

% ---------------------------------------------------
\section{Physical Axioms}
\label{sec:axioms}

The Information Theory of Action is based on a minimal set of physical assumptions. These axioms make no reference to wave functions, operators, Hilbert spaces, trajectories, amplitudes, or measurement postulates. They specify only the kinematical structure of physical processes in terms of configurations and action, cleanly separating physical content from inferential rules.

\begin{axiom}[Configuration space]
\label{ax:config}
A physical system is associated with a configuration space $\mathcal{Q}$, which labels the possible configurations of the system.
The space $\mathcal{Q}$ may be continuous, discrete, or hybrid, depending on the system under consideration.
\end{axiom}

The configuration space plays a purely kinematical role. At this stage, no probabilistic or quantum structure is assumed.

\begin{axiom}[Action space]
\label{ax:action}
A physical system is associated with an action space $\mathcal{A} \subseteq \mathbb{R}$, representing the possible values of total action accumulated during the evolution between configurations.
\end{axiom}

Action is taken as the fundamental dynamical quantity. It is a scalar, additive under sequential composition, independently of any underlying trajectory description, and central to classical mechanics, field theory, and variational principles. This motivates its role as the primitive dynamical variable of the theory.

\begin{axiom}[Density of action states]
\label{ax:density}
For any pair of configurations $a,b \in \mathcal{Q}$ and any time interval $T$, there exists a unique nonnegative function $g(A;b,T|a)$ such that:
\begin{enumerate}[label=(\alph*)]
\item \textbf{Positivity and integrability:}
\begin{equation}
g(A;b,T|a) \ge 0, \qquad \int_{\mathcal{A}} g(A;b,T|a)\, dA < \infty .
\end{equation}

\item \textbf{Composition from action additivity:}
for any intermediate configuration $b$ and times $T_1,T_2$,
\begin{eqnarray}
g(A;c,T_1+T_2|a)
&=&
\int_{\mathcal{Q}} db
\int_{\mathcal{A}} dA'\,
g(A';b,T_1|a)\,\nonumber\\ \times&&
g(A-A';c,T_2|b).
\label{eq:composition_g}
\end{eqnarray}

\item \textbf{Time-reversal symmetry:}
for systems without external time-dependent driving,
\begin{equation}
g(A;b,T|a) = g(A;a,T|b).
\label{eq:time_reversal}
\end{equation}

\item \textbf{Locality:}
for sufficiently small $T$, the density $g(A;b,T|a)$ is negligible unless $|b-a| \lesssim v_{\max} T$, where $v_{\max}$ is a characteristic velocity scale.

\item \textbf{Finite variance:}
\begin{equation}
\int_{\mathcal{A}} A^2\, g(A;b,T|a)\, dA < \infty .
\end{equation}
\end{enumerate}
\end{axiom}

The density $g(A;b,T|a)$ encodes the multiplicity of dynamical possibilities connecting $a$ to $b$ with total action $A$ over time $T$. Its defining property is the convolution structure \eqref{eq:composition_g}, which directly reflects the fundamental additivity of action under sequential processes.
It is essential to emphasize that $g(A;b,T|a)$ is \emph{not} a probability density and carries no epistemic meaning. It is a structural object encoding multiplicity, analogous to a density of states in statistical mechanics. Probabilities arise only after inference is performed relative to $g$.

\begin{remark}[Implicit definition and non-circularity]
The density $g(A;b,T|a)$ is defined \emph{implicitly} by its structural properties, rather than through an explicit construction in terms of trajectories or path integrals. This mirrors the role of fundamental objects in other theories, such as the density of states in statistical mechanics, which is characterized by axiomatic properties rather than microscopic constructions. In particular, no classical action, Lagrangian, or path space is assumed. We show later that the Lagrangian, the Euler--Lagrange equations, and the classical limit $g \to \delta(A-S_{\mathrm{cl}})$ emerge as consequences of these axioms, eliminating any circularity.
\end{remark}

When the time interval is clear from context, we write
$g(A;b|a) \equiv g(A;b,T|a)$.

% ---------------------------------------------------
\subsection{Existence, Uniqueness, and Absence of Circularity}
\label{sec:existence_uniqueness}

A natural concern is whether the density of action states $g(A;b|a)$, defined axiomatically in Axiom~\ref{ax:density}, actually exists, whether it is unique, and whether its definition implicitly assumes a path-integral structure. The function $g(A;b|a)$ is defined solely through abstract physical properties, positivity, action additivity, locality, and finite variance, without reference to trajectories, amplitudes, or Lagrangians.

\paragraph{Existence.}
For any system admitting a well-defined short-time action distribution consistent with locality and finite variance, a finite-time density $g(A;b,T|a)$ exists and is constructed uniquely by repeated composition. This follows from standard results on convolution semigroups and is outlined in Appendix~\ref{app:existence_uniqueness}.

\paragraph{Uniqueness.}
Once the short-time density is fixed, the composition law determines $g(A;b,T|a)$ uniquely for all times. Different physical systems correspond to different short-time means $\bar A(b,\Delta t|a)$, while the functional form of $g$ is fixed by the axioms up to action-space gauge equivalence.

\paragraph{No circularity.}
Crucially, no reference to paths or to the Feynman integral is made at this stage. The path-integral representation emerges later as a \emph{theorem}, when a trajectory space exists, rather than being assumed. The logical structure is, therefore:
\begin{equation}
\text{Axioms for } g
\;\longrightarrow\;
\text{Propagator } K
\;\longrightarrow\;
\underbrace{\text{Path integral}}_{\text{if applicable}} .
\end{equation}

All technical details are deferred to the Appendix~\ref{app:existence_uniqueness}.

\paragraph{Trajectory representations.}
The density of action states $g(A;b|a)$ is defined axiomatically, without reference to trajectories or path integrals. Nevertheless, when a system admits a trajectory space $\Gamma_{a\to b}$ and an action functional $S[\gamma]$, one can show that $g$ admits a representation as the Jacobian of the transformation from paths to action values. In this case, the propagator $K(b|a)$ reduces to the Feynman path integral as a derived result. Importantly, this representation is \emph{not assumed} but follows from the axioms together with the existence of paths. The complete theorem and proof are given in Appendix~\ref{app:jacobian}.

% ---------------------------------------------------
\subsection{Why the Density of Action States Is Defined Implicitly}
\label{sec:why_implicit}

A theory intended to reconstruct quantum mechanics should avoid introducing classical structures at the foundational level. For this reason, the density of action states $g(A;b|a)$ is defined implicitly through its physical and mathematical properties (Axiom~\ref{ax:density}), rather than by an explicit construction.

In particular, quantum systems do not possess fundamental trajectories. While a representation of $g$ in terms of paths exists when a suitable trajectory space can be defined, this representation is a derived result, not a primitive assumption. Treating paths as fundamental would reintroduce classical concepts at the axiomatic level, contrary to the aim of identifying the minimal structure underlying quantum mechanics.

A further motivation is mathematical. The functional measure appearing in the Feynman path integral is not rigorously defined in general~\cite{Simon2005,GlimmJaffe1987,ReedSimon1975}, and therefore cannot serve as a reliable primitive object. By contrast, the implicit definition of $g$ via positivity, composition, locality, and finite variance provides a mathematically well-posed starting point that does not rely on ill-defined constructions.

The implicit approach also ensures maximal generality. It applies equally to systems with continuous or discrete configuration spaces, to systems with or without a classical trajectory interpretation, and to cases where the notion of a path is ambiguous or absent altogether. When trajectories do exist, the path-integral representation emerges as a theorem, establishing that the axiomatic $g$ coincides with the Jacobian of the transformation from paths to action values.

Finally, the axioms uniquely determine $g$ for systems with well-defined short-time structure. Existence, uniqueness, and the absence of circularity are established rigorously in Appendix~A. The logical order of the theory is therefore unambiguous: the density of action states is fundamental, while trajectories and the Feynman path integral appear only as derived constructs when additional structure is available.

% ---------------------------------------------------
\subsection{Physical content and observability of $g$}
\label{sec:g_physical}

The density of action states $g(A; b|a)$ is not an observable quantity and should not be interpreted as a probability density. Its role is structural rather than epistemic: it encodes the multiplicity of dynamically distinct ways by which a system may evolve from configuration $a$ to $b$ while accumulating total action $A$.
The density $g(A;b|a)$ is therefore a structural object
encoding multiplicity, analogous to a density of states in
statistical mechanics, and serves as the reference measure
for inference rather than a probability distribution.
It is not normalized, need not be monotonic, and does not encode epistemic uncertainty. Instead, it provides the natural reference measure with respect to which inference is performed.

The physical content of $g$ becomes apparent once its role in the propagator is recognized. All observable transition amplitudes are determined by the propagator, $K(b|a)$, %the Fourier-type relation
% \begin{equation}
% K(b|a) = \int dA\, g(A;b|a)\, e^{i\eta A},\label{eq:K_equivalence}
%\end{equation}
so that $g$ governs the relative weight of different action contributions to physical processes. In this sense, $g$ is no less physical than the Lagrangian or the Hamiltonian: it is not directly observable by itself, but it fully determines observable transition probabilities through the propagator.

For a given physical system, the structure of $g(A;b|a)$ is not arbitrary. When a Lagrangian description exists, the axioms uniquely determine $g$, and, when a trajectory space can be defined, $g$ admits a concrete representation as the Jacobian of the map from paths to action values. More generally, symmetries of the system strongly constrain its form. For example, spatial translation invariance implies that $g$ depends only on the displacement $b-a$, while time-translation invariance fixes its dependence on the time interval.

In principle, $g$ could be reconstructed experimentally from precise knowledge of the propagator via an inverse Fourier transform, although this is practically challenging. More importantly, its short-time structure is highly constrained by the composition law and locality, which determine $g$ up to normalization in the infinitesimal time limit. Thus, $g$ should be viewed not as an auxiliary mathematical device, but as a compact encoding of the system's dynamical possibilities, from which both quantum and classical behavior ultimately emerge.

\paragraph*{Minimality of the axioms.}
It is important to emphasize that the axioms introduced above are deliberately minimal. They do \emph{not} assume amplitudes, complex numbers, phases, linear superposition, Hilbert spaces, unitarity, the Born rule, collapse postulates, trajectories, path integrals, or specific dynamical equations such as the Schrödinger equation. All of these structures will emerge later as \emph{derived consequences} of three ingredients: the multiplicity of action values encoded in $g(A;b|a)$, information-theoretic inference via the maximum entropy principle, and the operational indistinguishability implied by the Cramér--Rao bound.

% ---------------------------------------------------
% \subsection{Equivalence classes and operational content}
% \label{sec:equivalence_classes}

% A natural question concerns the uniqueness and physical status of the density of action states $g(A; b|a)$. Different microscopic realizations of $g$ may, in principle, lead to the same propagator (\ref{eq:K_equivalence}). Such realizations are physically equivalent, as all observable predictions depend solely on $K(b|a)$.

% Accordingly, the theory admits equivalence classes of action densities related by transformations that leave the propagator invariant:
% \begin{equation}
% g_1 \sim g_2 \quad \Leftrightarrow \quad \int dA\, g_1(A)\, e^{i\eta A} = \int dA\, g_2(A)\, e^{i\eta A}.
% \label{eq:equivalence_relation}
% \end{equation}
% This redundancy is analogous to gauge freedom in field theory: it reflects the fact that $g$ itself is not observable, while the propagator is. The physically meaningful content of the theory resides in the equivalence class rather than in a specific representative.

% ---------------------------------------------------

\section{Maximum Entropy on the Joint Space $(A,b)$}
\label{sec:maxent}

With the physical framework established, we now turn to the epistemic question: given incomplete information about which dynamical alternative occurred, how should probabilities be assigned? For a fixed initial configuration $a$, the elementary alternatives of the dynamics are pairs $(A,b)$, where $A\in\mathcal{A}$ is the total action accumulated and $b\in\mathcal{Q}$ is the final configuration. Our state of knowledge about the process is therefore described by a joint probability density $P(A,b|a)$ on $\mathcal{A}\times\mathcal{Q}$.

The density of action states $g(A;b|a)$ encodes the physical multiplicity of dynamical possibilities and plays the role of a reference measure on this space. Given $g$, rational inference under incomplete information requires assigning $P(A,b|a)$ by maximizing the relative entropy subject to the available constraints, following the Jaynes–Shore–Johnson principle, see Appendix~\ref{app:maxent}.

Assuming only normalization and a fixed mean action,
\[
\langle A\rangle = \int db\, dA\, A\, P(A,b|a),
\]
the maximum-entropy distribution takes the universal exponential form
\begin{equation}
P(A,b|a)
=
\frac{g(A;b|a)\, e^{-\eta A}}{Z(\eta;a)},
\label{eq:MaxEnt_joint}
\end{equation}
where $\eta$ is a Lagrange multiplier conjugate to the action and
\[
Z(\eta;a)=\int db\, dA\, g(A;b|a)\, e^{-\eta A}.
\]
At this stage $\eta$ has purely inferential meaning as a
Lagrange multiplier enforcing the mean-action constraint.
Its physical interpretation will arise only after the
propagator structure is established.
Here, $\eta$ quantifies how sharply different action values can be resolved and will be shown in the next section to define a fundamental action-resolution scale through the Cramér–Rao bound. This action-resolution scale will be shown in the next section to imply a universal lower bound on distinguishability, providing the operational origin of quantum coherence.

% ---------------------------------------------------
\section{Cramér--Rao Bound and Indistinguishability of Action}
\label{sec:CR}

The maximum-entropy distribution on the joint space $(A,b)$ introduced a real parameter $\eta$ conjugate to the action, in Eq. (\ref{eq:MaxEnt_joint}).  At this stage, $\eta$ is purely an information-theoretic quantity, it enforces the constraint on the mean action.

A fundamental question then arises:
\emph{what is the operational meaning of $\eta$, and how finely can action be resolved?}

To answer this question in this section we rely on a minimal and explicit set of assumptions: (i) locality of the dynamics, (ii) additivity of action under sequential composition, (iii) finite second moment of the short-time action distribution, and (iv) absence of singular or distributional Lagrangians. These conditions are satisfied by ordinary nonrelativistic mechanical systems and define the regime of applicability of the present reconstruction.

\medskip

The marginal distribution over action, $P(A|a;\eta)=\int db\,P(A,b|a;\eta)$, belongs to an exponential family. For such families, the Fisher information associated with $\eta$ is equal to the variance of the conjugate variable $A$. As a consequence, the Cramér--Rao bound~\cite{Cramer1946,amari2000methods} implies the inequality
\begin{equation}
\Delta A\,\Delta\eta \ge 1 ,
\label{eq:CR_main}
\end{equation}
which sets a fundamental limit on the simultaneous resolution of the action and its conjugate parameter. This bound is purely informational and does not rely on quantum postulates.

\medskip

Independently, the propagator connecting configurations $a$ and $b$ to take the form of $K(b|a) = \int_{\mathcal{A}} dA\, g(A;b|a)\, e^{i\eta A}$, see section (\ref{sec:propagator}), which is a Fourier transform with respect to the action variable. This establishes a mathematical conjugacy between $A$ and $\eta$, implying an intrinsic Fourier-duality-induced uncertainty between their characteristic scales. The Cramér--Rao bound \eqref{eq:CR_main} is therefore not accidental: it reflects the deep transform structure of the theory.

\medskip

For $\eta$ to play a physically meaningful role, it must modulate the distribution $P(A)$ over the natural support of $g(A)$. If $\eta\,\Delta A\ll1$, the exponential factor is irrelevant; if $\eta\,\Delta A\gg1$, it suppresses all structure in $g$. Physical relevance therefore requires $\eta\,\Delta A=O(1)$. Combining this naturalness condition with \eqref{eq:CR_main} yields
\begin{equation}
\Delta A_{\min}\sim \frac{1}{\eta}.
\label{eq:action_resolution}
\end{equation}
This result has a direct operational interpretation: action differences smaller than $1/\eta$ are fundamentally indistinguishable by any inference consistent with the maximum-entropy principle. Distinguishability is therefore not absolute but limited by information.

\medskip

When a trajectory representation exists, two paths with actions $A_1$ and $A_2$ are operationally distinguishable only if $|A_1-A_2|\gg1/\eta$. If $|A_1-A_2|\lesssim1/\eta$, no measurement can determine which alternative occurred. Such alternatives must therefore be combined coherently.

\medskip

Crucially, for infinitesimal time segments, action differences scale linearly with $dt$, while the minimal resolvable scale remains finite. As a result, \emph{all infinitesimal segments are universally indistinguishable}, see Appendix~\ref{sec:CR_infinitesimal}. This leads naturally to coherent (amplitude-level) composition at every step, making a coherent amplitude-level description the minimal
consistent representation of infinitesimal composition.

\medskip

%Under the assumptions stated above, the Cram\'er--Rao bound implies that the minimal resolvable action scale remains finite as $dt \to 0$,
%while action differences between infinitesimal segments scale linearly with $dt$. As a consequence, all infinitesimal contributions are operationally indistinguishable. This result is kinematical rather than dynamical: it follows from scaling and inference, independently of the specific form of the Lagrangian.
In summary, inference theory imposes a finite, time-independent resolution scale in action. Because infinitesimal action increments scale linearly with $dt$, while the resolution scale does not, all infinitesimal dynamical alternatives become operationally indistinguishable. This result is purely kinematical and independent of the specific form of the dynamics. Fourier-duality-induced uncertainty

In Sec.~\ref{sec:hbar} we will show that experimental calibration identifies $\eta=1/\hbar$. The fundamental quantum of action thus emerges as an information-theoretic resolution scale, not as an independent postulate.

% ---------------------------------------------------
\section{From Indistinguishability to Complex Amplitudes}
\label{sec:amplitudes}

The Cramér--Rao bound derived in Sec.~\ref{sec:CR} establishes a finite resolution scale in action space; two processes with action values satisfying $|A_1 - A_2| \lesssim 1/\eta$ are operationally indistinguishable and cannot be assigned independent classical probabilities. Furthermore, Theorem~\ref{th:infinitesimal_indist} in Appendix~\ref{app:indistinguishability} shows that this indistinguishability is \emph{universal} at the infinitesimal level, for any two trajectories with finite Lagrangians, the ratio $|\delta S_1 - \delta S_2|/\Delta(\delta S)_{\min}$ vanishes as $dt \to 0$, where $S_{i}$ is the trajectory action. This means that any consistent dynamical theory must specify how indistinguishable alternatives combine when composing infinitesimal time steps.

A purely classical rule, $P_{\text{tot}} = P_1 + P_2$, is \emph{logically excluded} once we require that probabilities remain normalized under dynamical evolution. To see this, consider that sequential composition of weights must respect action additivity: $\chi(A_1 + A_2) = \chi(A_1) \otimes \chi(A_2)$. For a continuous character $\chi: \mathbb{R} \to \Phi$, this is the Cauchy functional equation. If $\Phi = \mathbb{R}_{\geq 0}$ with ordinary multiplication, the only solutions are $\chi(A) = e^{\lambda A}$ for real $\lambda$. However, \emph{probability normalization}---the requirement that $\sum_i P_i = 1$ be preserved under evolution---combined with the coherent combination rule forces $|\chi(A)|$ to remain bounded for all $A$. For $\chi(A) = e^{\lambda A}$: if $\lambda > 0$, probabilities grow without bound (violating normalization); if $\lambda < 0$, probabilities decay irreversibly to zero (information loss); if $\lambda = 0$, dynamics is trivial. None yields a consistent probabilistic theory. The resolution is $\Phi = \mathbb{C}$ with $\chi(A) = e^{i\eta A}$: bounded ($|\chi| = 1$), invertible ($\chi(-A) = \chi(A)^{-1}$), and non-trivial. Thus, unitarity $|\chi(A)| = 1$ is not an independent axiom but a consequence of probability normalization in a theory with coherent combination of indistinguishable alternatives. Complex amplitudes and interference emerge as consequences
of indistinguishability combined with the requirement that
probabilities remain normalized within the present framework.

To formalize this, we introduce a \emph{weight} $\phi(A) \in \Phi$ associated with each alternative labeled by its action value, and impose minimal consistency requirements: (i) indistinguishable alternatives combine via an associative, commutative operation $\oplus$; (ii) sequential composition respects action additivity, so that $\phi(A_1 + A_2) = \phi(A_1) \otimes \phi(A_2)$; (iii) observable probabilities are extracted by a nonnegative map $f: \Phi \to [0,\infty)$; (iv) for distinguishable alternatives, probabilities add classically: $f(\phi_1 \oplus \phi_2) = f(\phi_1) + f(\phi_2)$; and (v) probability normalization: total probability $\sum_i P_i = 1$ must be preserved under dynamical evolution, with non-trivial dynamics ($\chi$ not identically constant). Under these conditions, the weight space is uniquely determined to be complex, unitarity $|\chi(A)| = 1$ follows as a theorem, and the minimal continuous realization follows.

\begin{theorem}[Complex amplitude representation]
\label{th:complex_amplitudes}
Assume (a) action additivity under sequential composition, (b) a well-defined combination rule for indistinguishable alternatives satisfying associativity, commutativity, and continuity, and (c) probability normalization: total probability $\sum_i P_i = 1$ is preserved under dynamical evolution, with $\chi$ not identically constant (non-trivial dynamics). Then:
\begin{enumerate}[label = (\alph*)]
    \item The minimal continuous representation of the weight space is $\Phi \simeq \mathbb{C}$;
    \item Indistinguishable alternatives combine via complex addition: $\oplus \leftrightarrow +$;
    \item Sequential composition corresponds to complex multiplication: $\otimes \leftrightarrow \cdot$;
    \item Probabilities are given by the squared modulus (Born rule):
    \begin{equation}
    P = |\phi_{\text{tot}}|^2;
    \label{eq:born_rule}
    \end{equation}
    \item Unitarity follows as a theorem: $|\chi(A)| = 1$ for all $A$.
\end{enumerate}
This result establishes the minimal continuous realization
compatible with action additivity, coherent composition,
and probability normalization within the present framework. 
Moreover, compatibility with action additivity implies that the weight function $\chi: A \mapsto \phi(A)$ is a continuous character of $(\mathbb{R}, +)$, uniquely determining its form,
\begin{equation}
\chi(A) = e^{i\eta A},
\label{eq:chi_form}
\end{equation}
for some real constant $\eta$.

\end{theorem}

Probability normalization plays no role as an additional physical postulate in this derivation. It is a definitional requirement for any probabilistic interpretation: a quantity that does not sum to unity cannot be interpreted as a probability. Once coherent combination of indistinguishable alternatives is enforced by the Cram\'er--Rao bound, normalization preservation uniquely constrains the admissible weight space, leading to complex amplitudes and the quadratic probability rule $P = |\phi|^2$. The complete proof is given in Appendix~\ref{app:complex_amplitudes}. 

Equation~\eqref{eq:chi_form} reveals a fundamental connection: the same parameter $\eta$ that controls action resolution (Sec.~\ref{sec:CR}) also governs phase sensitivity. Phase differences become significant when $|\eta(A_1 - A_2)| \gtrsim 1$, which matches precisely the indistinguishability window $|A_1 - A_2| \lesssim 1/\eta$. This is not a coincidence, it reflects the deep connection between action resolution and quantum coherence. The transition from the real exponential $e^{-\eta A}$ in the MaxEnt distribution~\eqref{eq:MaxEnt_joint} to the complex phase $e^{i\eta A}$ is the mathematical signature of moving from classical probabilistic description (distinguishable alternatives, incoherent addition) to quantum description (indistinguishable alternatives, coherent addition as complex amplitudes).

\medskip
The algebraic derivation that consistency requirements force $\Phi = \mathbb{C}$ follows a structure similar to Goyal et al.~\cite{Goyal2010}. However, a crucial difference distinguishes ITA: while Goyal et al.\ require interference as an empirical postulate, ITA derives complex amplitudes from more primitive
requirements within the present inferential framework. Condition (d) above, probability normalization, is not a physical assumption about dynamics; it is a \emph{definitional requirement} for any consistent probabilistic theory. Probabilities that do not sum to unity are not probabilities. Combined with coherent combination of indistinguishable alternatives (forced by the Cram\'er--Rao bound, Theorem~\ref{th:infinitesimal_indist}), normalization preservation uniquely determines $\Phi = \mathbb{C}$ and unitarity $|\chi| = 1$. Interference then \emph{emerges} as a prediction, not an input. The scale $\eta$ emerges from the Cram\'er--Rao bound rather than being an independent parameter.

Now we pose the following question:
\textit{What if probability normalization were violated?}
Probability normalization is not a physical assumption that could be ``violated'', it is part of the definition of probability. A ``probability'' that does not sum to unity is simply not a probability in the Kolmogorov sense. The question ``what if normalization failed?'' is therefore ill-posed within probability theory. However, we can ask: \textit{what if we attempted to build a theory with distinguishable alternatives (incoherent combination) instead of indistinguishable ones?} In that case, $\Phi = \mathbb{R}_{\geq 0}$ would suffice, and the result would be classical stochastic mechanics: no quantum interference, no complex amplitudes. The key insight is that indistinguishability (from Cram\'er--Rao) \emph{forces} coherent combination, and coherent combination plus normalization \emph{forces} unitarity. Once $\Phi = \mathbb{C}$ and $|\chi| = 1$ are established, destructive and constructive interference follow automatically from the geometry of complex addition.

%  cccccccccccccccccccccccccccccccccccccccccccccccccccccccccccccccccccccccccccccccccccccccccccccccccccccccccccccccccccccccccccccccccccccccccccccccccccccccccccccccccccccccccccccccccccccccccccccccccccccccccccccccccccccccccccccccccccccccc\textit{Coherence, incoherence, and amplitude space.}
It is important to clarify that the necessity of complex amplitudes should not be conflated with the universal observability of interference.
In systems with constraints, singular Lagrangians, topological terms, or superselection sectors, interference between certain alternatives
may be operationally suppressed. In such cases, contributions combine effectively incoherently, leading to classical probability addition.
However, this loss of coherence does not signal a reduction of the underlying amplitude space from $\mathbb{C}$ to $\mathbb{R}$. Rather, it reflects a loss of operational phase resolution. The amplitudes remain complex, but their relative phases become inaccessible due to dynamical, inferential, or topological reasons. Real-valued probabilities then emerge only after phase information is effectively averaged out or rendered unobservable. Thus, incoherent summation corresponds to a regime of restricted phase accessibility, not to a change in the structural representation of dynamical alternatives. The complex character of amplitudes remains the minimal and universal framework compatible with action additivity and probability normalization.

% ---------------------------------------------------
\section{Propagator as a Derived Object}
\label{sec:propagator}
The propagator is obtained by combining the action multiplicity with the continuous character of the additive action group:
\begin{equation}
K(b|a)=\int_{\mathcal{A}} dA\, g(A;b|a)\,\chi(A).
\label{eq:K_chi_A1}
\end{equation}
Once indistinguishable alternatives are shown to combine coherently with the character $\chi(A) = e^{i\eta A}$, the propagator follows uniquely. Consistency with action additivity and the composition law for $g(A;b|a)$ requires that any kernel $\mathcal{K}(b|a)$ satisfying (i) linearity in $g$, (ii) sequential composition, and (iii) correct phase assignment for definite-action contributions must have the form
\begin{equation}
K(b|a) = \int_{\mathcal{A}} dA\, g(A;b|a)\, e^{i\eta A}.
\label{eq:K_def}
\end{equation}
%consistent with Eq. (\ref{eq:K_equivalence}).
This expression is \emph{not} postulated: it is the natural kernel compatible with the character structure derived in Section~\ref{sec:amplitudes}. The propagator therefore emerges as a \textit{derived object}, not a new axiom. The density $g(A;b|a)$ itself is not uniquely fixed by physical observables; only the propagator $K(b|a)$ has direct physical meaning. Different representations of $g$ related by action-space transformations lead to identical predictions.

A crucial property follows immediately from action additivity and the composition law~\eqref{eq:composition_g}: the propagator satisfies the \emph{semigroup property}
\begin{equation}
K(c|a) = \int_{\mathcal{Q}} db\, K(c|b)\, K(b|a),
\label{eq:semigroup}
\end{equation}
which states that evolution from $a$ to $c$ can be decomposed through any intermediate configuration $b$. This property is essential for the emergence of unitary time evolution and the Schrödinger equation in later sections.

The detailed proofs of uniqueness and the semigroup property are given in Appendix~\ref{app:propagator}. Note that the action-state density $g(A;b,T|a)$ is a representational object. Distinct densities related by action-space transformations may lead to the same propagator $K(b|a)$, and therefore to identical physical predictions. As in other gauge theories, only the propagator is physically meaningful.

% ---------------------------------------------------
\section{Indistinguishability Scale and Identification with $\hbar$}
\label{sec:hbar}

Having derived the propagator structure, we now address the physical scale that governs quantum behavior. The MaxEnt distribution, the Cramér--Rao bound, and the phase character introduce a single action scale. From Sec.~\ref{sec:CR}, statistical inference on action implies an operational resolution $\Delta A_{\min}\sim 1/\eta$, i.e., action differences below this scale cannot be reliably discriminated.
Indistinguishability then requires coherent weights (Sec.~\ref{sec:amplitudes}), and the unique continuous character compatible with action additivity is $\chi(A)=e^{i\eta A}$. Thus the same parameter $\eta$ controls both (i) the action scale relevant to inference and (ii) the rate of phase accumulation with action. The detailed justification of why a single $\eta$ unifies these three contexts is given in Appendix~\ref{app:scale_unification}.

The present framework determines the structural role of a universal action scale but not its numerical value. The identification $\eta = 1/\hbar$ is therefore \emph{empirical}: it is fixed by matching the predicted phase sensitivity to observed interference phenomena (see Appendix~\ref{app:double_slit} for a worked example using double-slit interference). In this sense, Information Theory of Action explains \emph{why} quantum mechanics requires a fundamental action scale, while its numerical value must be calibrated experimentally,
exactly as in standard quantum mechanics. This calibration yields
\begin{equation}
\eta=\frac{1}{\hbar}.
\label{eq:eta_hbar}
\end{equation}
It is convenient to introduce the dimensionless \emph{action number}
\begin{equation}
\mathcal{N}\equiv \eta A=\frac{A}{\hbar},
\label{eq:action_number}
\end{equation}
which controls the quantum--classical regime: interference is relevant for $\mathcal{N}\lesssim O(1)$, while classical behavior emerges for $\mathcal{N}\gg 1$. As a concrete example, a mass $m = 1\,\text{g}$ moving at $v = 1\,\text{m/s}$ for $t = 1\,\text{s}$ has action $A \sim mvL \sim mv^2 t \sim 10^{-3}\,\text{J}\cdot\text{s}$, giving $\mathcal{N} \sim 10^{31}$, deep in the classical regime where interference
effects are completely negligible.

With~\eqref{eq:eta_hbar}, the propagator reads
\begin{equation}
K(b|a)=\int_{\mathcal A} dA\, g(A;b|a)\, e^{iA/\hbar},
\label{eq:K_physical}
\end{equation}
recovering the standard phase structure of the Feynman kernel as a derived output.

It is important to stress that the present theory does not aim to predict the numerical value of $\hbar$. Rather, it explains why a universal action scale must exist and why it enters quantum dynamics in the specific combination $A/\hbar$. The status of $\hbar$ is therefore analogous to that of coupling constants in classical field theories: its existence is theoretically necessary, while its value is empirically determined. The universality of the calibrated value of $\eta$ across all physical systems reflects the empirical universality of quantum interference, justifying the identification of a single constant $\hbar$ rather than system-dependent action scales. With the physical scale fixed, we are now in a position to analyze the short-time structure of $g(A;b,T|a)$.

% ---------------------------------------------------
\section{Short-Time Structure of $g$ from the Composition Law}
\label{sec:shorttime_g}

The propagator formula~\eqref{eq:K_physical} requires knowledge of the density $g(A;b|a)$. A central result of the theory is that the short-time form of the action density $g(A;b,\Delta t|a)$ is not postulated from quantum mechanics: it is uniquely determined by the composition law together with standard physical regularity assumptions. Integrating over intermediate configurations, the composition law implies that
the \emph{marginal} action density $\tilde{g}_T(A) \equiv \int_{\mathcal{Q}} db\, g(A;b,T|a)$ forms a \emph{convolution semigroup} in $A$:
\begin{equation}
\tilde{g}_{T_1+T_2} = \tilde{g}_{T_1} * \tilde{g}_{T_2}.
\label{eq:semigroup_action}
\end{equation}
Under the physically necessary requirement of \emph{finite variance} and the continuity condition $\tilde{g}_T \Rightarrow \delta$ as $T \to 0^+$, the Lévy--Khintchine classification~\cite{Feller1971,Sato1999} reduces the semigroup to the Gaussian class (proof in Appendix~\ref{app:shorttime_derivation}).

\begin{proposition}[Gaussian short-time fluctuations]
\label{prop:gaussian_short_time}
Let $\tilde{g}_T$ satisfy~\eqref{eq:semigroup_action} with finite second moment and $\tilde{g}_T \Rightarrow \delta$ as $T \to 0^+$. Then
\begin{equation}
\tilde{g}_T(A) = \frac{1}{\sqrt{2\pi\sigma^2 T}} \exp\!\left[-\frac{(A - \mu T)^2}{2\sigma^2 T}\right],
\label{eq:gaussian_marginal}
\end{equation}
for constants $\sigma^2 > 0$ and $\mu \in \mathbb{R}$.
\end{proposition}

The Gaussian form is not a quantum assumption but a mathematical consequence within the class of finite-variance convolution semigroups.

For a given physical system, the endpoint dependence shifts the mean to the \emph{mean action} $\bar{A}(b,\Delta t|a)$, which defines the emergent Lagrangian via $\bar{A} = L(x,v)\Delta t + O(\Delta t^2)$ in the local limit. Thus the universal content is the Gaussian \emph{fluctuation} structure around $\bar{A}$, whereas the system-specific content is encoded in $L$ (equivalently, in the potential $V$). This mirrors standard quantum mechanics -- the framework fixes the kinematics of evolution, while the specific interaction must be supplied.

The width parameter $\sigma$ is constrained by the indistinguishability scale: the spread in action must match the minimum resolvable increment $\Delta A_{\min} \sim 1/\eta = \hbar$, giving $\sigma\sqrt{\Delta t} \sim \hbar$ (see Appendix~\ref{app:variance_scaling} for details). For Galilean-invariant systems, symmetry requirements further constrain the Lagrangian to the form $L = \tfrac{1}{2}mv^2 - V(x)$, where the kinetic term is universal and only the potential $V(x)$ is system-specific input (Appendix~\ref{app:galilean_lagrangian}).

The resulting short-time propagator is
\begin{equation}
K(x_b, \Delta t | x_a) = \left(\frac{m}{2\pi i\hbar\Delta t}\right)^{d/2} e^{\!\left[\frac{i}{\hbar}\left(\frac{m(x_b - x_a)^2}{2\Delta t}
- V(\bar{x})\Delta t\right)\right]},
\label{eq:K_shorttime_final}
\end{equation}
which is the standard quantum mechanical form, now derived from the axioms rather than assumed. The normalization factor $(m/2\pi i\hbar\Delta t)^{d/2}$ follows from composition law consistency without invoking quantum mechanical concepts (Appendix~\ref{app:shorttime_derivation}).

% ---------------------------------------------------
\section{Emergence of the Lagrangian and the Variational Principle}
\label{sec:lagrangian_emergence}

The short-time kernel derived above contains the mean action $\bar{A}(b,\Delta t|a)$, which defines the system-specific dynamics. A central result of ITA is that the Lagrangian and the variational principle are not inputs but \emph{derived consequences} of the composition law and stationary phase analysis. This section establishes this emergence and shows that the classical limit $g \to \delta(A - S_{\text{cl}})$ is a theorem, not an assumption.

\subsection{Mean action and emergent Lagrangian}

The \emph{mean action} for the transition $a \to b$ in time $T$ is $\bar{A}(b,T|a) \equiv \int dA\, A\, g(A;b,T|a)$.
The variance $\sigma^2$ scales linearly with time, $\sigma^2 \propto T$, as established in Appendix~\ref{app:variance_scaling}.

\begin{definition}[Emergent Lagrangian]
\label{def:emergent_lagrangian}
The \emph{Lagrangian} is defined as the rate of mean action accumulation:
\begin{equation}
L(x, v) \equiv \lim_{dt \to 0} \frac{\bar{A}(x + v\,dt,\, dt\, |\, x)}{dt}
\label{eq:lagrangian_emergent}
\end{equation}
where $v = (x_b - x_a)/dt$ is the velocity.
\end{definition}

This definition inverts the usual logic order: instead of specifying $L$ and computing amplitudes, we specify $g(A;b|a)$ and \emph{derive} the Lagrangian from it. Different physical systems correspond to different functions $g$, which determine different Lagrangians.

\medskip
\noindent\textbf{System specification.}
A natural question arises; how is the mean action $\bar{A}(b,\Delta t|a)$ determined for a specific physical system? The answer parallels standard quantum mechanics: the kinetic structure $T = \tfrac{1}{2}mv^2$ follows universally from Galilean invariance (see Appendix~\ref{app:galilean_lagrangian}), while the potential $V(x)$ must be specified as physical input. For a harmonic oscillator, one sets $V = \tfrac{1}{2}m\omega^2 x^2$; for a Coulomb system, $V = -e^2/(4\pi\epsilon_0 r)$. This reflects the same logical structure present in both
classical and quantum mechanics--in standard QM one must also specify $\hat{H}$, and in classical mechanics one specifies $L$. The contribution of ITA is to show that, \emph{given} such a specification via $\bar{A}$, the quantum propagator, interference, and unitary evolution follow necessarily from action indistinguishability and probability normalization. Different potentials yield different densities $g(A;b|a)$ peaked at different classical action values, but the \emph{structure} of quantum mechanics (complex amplitudes, composition law, Hilbert space) is universal.

\medskip
\noindent\textbf{Bidirectional consistency.}
The relationship between $g(A;b|a)$ and the Lagrangian $L$ is bidirectional, not circular. Given $g$, the Lagrangian emerges via Definition~\ref{def:emergent_lagrangian}: $L = \lim_{dt\to 0} \bar{A}/dt$. Conversely, given $L$, the short-time density is constructed as a Gaussian centered on the classical action:
\begin{equation}
g(A; b, \Delta t | a) = \frac{1}{\sqrt{2\pi\sigma^2\Delta t}} \exp\left[-\frac{(A - L(\bar{x}, v)\Delta t)^2}{2\sigma^2\Delta t}\right],
\label{eq:g_from_L}
\end{equation}
where $\bar{x} = (a+b)/2$, $v = (b-a)/\Delta t$, and $\sigma^2 \sim \hbar^2/\Delta t$ from Cram\'er--Rao consistency. This bidirectionality is a \emph{consistency check}, not a logical defect: the two directions must agree for the theory to be self-consistent. A worked example for the free particle is given in Appendix~\ref{app:free_particle}.

\subsection{Stationary phase and the variational principle}

In the classical regime $\eta\bar{A} \gg 1$, the propagator takes the form $K(b|a) \approx \mathbb{N}(b,a)\, e^{i\eta\bar{A}(b|a)}$, where
$\mathbb{N}$ is a slowly varying prefactor. The composition law then becomes dominated by stationary phase contributions.

\begin{theorem}[Stationary phase condition]
\label{th:stationary_phase}
In the classical regime, the propagator composition $K(c, T_1+T_2 | a) = \int db\, K(c, T_2 | b)\, K(b, T_1 | a)$ is dominated by the intermediate configuration $b^*$ satisfying:
\begin{equation}
\frac{\partial \bar{A}(b, T_1 | a)}{\partial b}\bigg|_{b^*} + \frac{\partial \bar{A}(c, T_2 | b)}{\partial b}\bigg|_{b^*} = 0
\label{eq:stationary_condition}
\end{equation}
\end{theorem}

The condition~\eqref{eq:stationary_condition} expresses momentum continuity: defining the canonical momentum $p \equiv \partial\bar{A}/\partial b$, we have $p_{\text{final}}^{(a \to b^*)} = p_{\text{initial}}^{(b^* \to c)}$.

\subsection{Euler-Lagrange equations from composition}

Iterating the stationary phase condition over a discretized path and taking the continuum limit yields the classical equations of motion.

\begin{theorem}[Euler-Lagrange from composition]
\label{th:euler_lagrange}
The trajectory $\{x^*(t)\}$ determined by the stationary phase condition satisfies:
\begin{equation}
\frac{d}{dt}\frac{\partial L}{\partial \dot{x}} - \frac{\partial L}{\partial x} = 0
\label{eq:euler_lagrange_derived}
\end{equation}
where $L$ is the emergent Lagrangian~\eqref{eq:lagrangian_emergent}.
\end{theorem}

The proof proceeds by discretizing time into $N$ steps, applying stationary phase at each intermediate point, expanding the short-time action
$\bar{A}(x_{k+1}, dt | x_k) = L(\bar{x}_k, v_k)\,dt + O(dt^2)$, and taking $N \to \infty$. The complete derivation is given in Appendix~\ref{app:euler_lagrange}.

\begin{theorem}[Classical action theorem]
\label{th:classical_action}
The mean action along the stationary trajectory equals the classical action:
\begin{equation}
\bar{A}(b,T|a)\big|_{\text{stationary}} = S_{\text{cl}}(b,T|a) \equiv \int_0^T L(x^*(t), \dot{x}^*(t))\, dt
\label{eq:classical_action_theorem}
\end{equation}
\end{theorem}

\begin{corollary}[Classical limit of $g$]
\label{cor:g_classical_limit}
In the classical regime $\mathcal{N} \to \infty$, the density becomes sharply peaked:
\begin{equation}
g(A;b,T|a) \xrightarrow{\mathcal{N} \to \infty} \delta(A - S_{\text{cl}}(b,T|a)).
\label{eq:g_to_delta}
\end{equation}
\end{corollary}

The classical limit thus corresponds to a large-deviation principle in action space, with the classical trajectory emerging as the most probable contribution.

\begin{proof}
From Cram\'er--Rao, the minimum action resolution is $\hbar$. As $\mathcal{N} \to \infty$, the mean action $\bar{A} = \mathcal{N}\hbar$ grows without bound while the width $\sigma$ grows more slowly, so $\sigma/\bar{A} \to 0$ and $g$ concentrates around $S_{\text{cl}}$.
\end{proof}

This corollary is significant -- the classical limit condition, which might have been assumed as an axiom, is now a \emph{derived theorem}. The Lagrangian, Euler-Lagrange equations, and classical action all emerge from the properties of $g$ without circular assumptions.

% ---------------------------------------------------
\section{Emergence of Hilbert Space Structure}
\label{sec:hilbert}

With the propagator and classical limit established, we now show that the standard mathematical framework of quantum mechanics emerges automatically. Once the propagator $K(b|a)$ is established as a complex-valued kernel satisfying the semigroup composition law and time-reversal symmetry, the standard Hilbert space formulation follows naturally. Here, the Hilbert space is not introduced as a fundamental arena for states, but as the minimal mathematical structure required to represent the propagator consistently. The propagator defines a linear evolution operator $U$ acting on square-integrable functions via $(U\psi)(b) = \int K(b|a)\psi(a)\,da$, equipped with the standard $L^2$ inner product $\langle\phi|\psi\rangle = \int \phi^*(q)\psi(q)\,dq$. Linearity follows directly from the additive structure of action and the coherent combination rule for indistinguishable alternatives, not from an independent postulate.

The key result is that this evolution is \emph{unitary}. For systems with real Lagrangians where kinetic energy is even in velocity, time-reversal symmetry implies $K^*(b,T|a) = K(a,-T|b)$. Combined with the composition law, this guarantees that inner products are preserved: $\langle U\phi|U\psi\rangle = \langle\phi|\psi\rangle$ (see Appendix~\ref{app:hilbert_structure} for the complete proof). Unitarity is therefore not an independent assumption about time evolution, but a consequence of time-reversal symmetry combined with the composition law.

At this stage, no additional physical assumptions are required; what follows is a mathematical consequence of the established properties of $U(t)$. Strong continuity of the one-parameter family $\{U(t)\}$ then ensures, by Stone's theorem, the existence of a self-adjoint Hamiltonian $\hat{H}$ such that $U(t)=e^{-i\hat{H}t/\hbar}$. At this stage, all physical assumptions have been exhausted; the emergence of $\hat{H}$ is a purely mathematical consequence of unitarity and continuity.

The Hilbert space $L^2(\mathcal{Q})$, the self-adjoint generator $\hat{H}$, and unitary dynamics therefore arise as a representational framework induced by the propagator. None of these structures are postulated at the fundamental level; they are consequences of action multiplicity, coherent combination of indistinguishable alternatives, and probability normalization.

%Strong continuity of the one-parameter family $\{U(t)\}$ then ensures, by Stone's theorem~\cite{Stone1932,ReedSimon1980}, the existence of a self-adjoint Hamiltonian $\hat{H}$ such that $U(t) = e^{-i\hat{H}t/\hbar}$. The entire Hilbert space structure (the function space $L^2(\mathcal{Q})$, the self-adjoint generator $\hat{H}$, and unitary dynamics) emerges from the propagator's algebraic properties, none of these being postulated; instead they are consequences of action multiplicity and coherent combination of amplitudes.

% ---------------------------------------------------
\section{Emergent Quantum Dynamics}
\label{sec:dynamics}

The dynamical equations of quantum mechanics emerge directly from the propagator structure derived in the preceding sections. The short-time kernel $K(x,\Delta t|x')$, obtained from the composition law and Galilean invariance (Section~\ref{sec:shorttime_g}), already encodes the complete dynamics, no additional dynamical postulates are introduced. In particular, no canonical quantization rule, operator correspondence, or commutation relations are assumed at any stage of the derivation. The kernel's Gaussian form with width $\sim\sqrt{\Delta t}$ and phase $\sim (x-x')^2/\Delta t$ is not borrowed from quantum mechanics but derived from the Lévy--Khintchine theorem and symmetry principles.

Once the short-time kernel is fixed, the continuum limit is purely kinematical and does not require additional physical assumptions. Taking the continuum limit $\Delta t \to 0$ in the evolution equation $\psi(x,t+\Delta t) = \int K(x,\Delta t|x')\psi(x',t)\,dx'$ yields the
Schrödinger equation (see Appendix~\ref{app:schrodinger_derivation} for details):
\begin{equation}
i\hbar\frac{\partial\psi}{\partial t} = \left(-\frac{\hbar^2}{2m}\nabla^2 + V(x)\right)\psi.
\label{eq:Schroedinger}
\end{equation}
Equation~\eqref{eq:Schroedinger} therefore encodes the effective description of coherent propagation, not a fundamental law imposed on the system.

The operator $\hat{H} = -\frac{\hbar^2}{2m}\nabla^2 + V(x)$ is precisely the quantization of the classical Hamiltonian $H = p^2/2m + V$, obtained via Legendre transform of the emergent Lagrangian $L = \frac{1}{2}m\dot{x}^2 - V(x)$. The momentum operator $\hat{p} = -i\hbar\nabla$ arises as the generator of spatial translations defined by the propagator, rather than as a quantized classical observable. The canonical commutation relation $[\hat{x},\hat{p}] = i\hbar$ follows automatically from this representation (Appendix~\ref{app:schrodinger_derivation}).

The conceptual significance is clear: the Schrödinger equation is not a fundamental postulate but the \emph{universal continuum limit} of coherent action-weighted propagation. The logical chain
\[
g(A;b|a) \xrightarrow{\substack{\text{MaxEnt} \\ \text{+ coherence}}} K(b|a)
\xrightarrow{\substack{\text{short-time} \\ \text{limit}}} i\hbar\partial_t\psi = \hat{H}\psi
\]
shows that quantum dynamics emerges from the information-theoretic structure encoded in the density of action states. The only system-specific input is the potential $V(x)$, which must be provided in any formulation, classical or quantum.

%This completes the derivation of standard quantum mechanics from ITA. Starting from the density $g(A;b|a)$, composition law, and physical symmetries, we have obtained: complex amplitudes, the propagator, Hilbert space structure, unitary evolution, the Schrödinger equation, and canonical commutation relations. %None of these were assumed; all %emerge as necessary %consequences of action %multiplicity and coherent %combination.

This completes the reconstruction of quantum dynamics within the Information Theory of Action. Quantum mechanics appears not as a theory of microscopic objects, but as the unique consistent framework for coherent inference over indistinguishable action alternatives.

\section{Conclusion}
\label{sec:conclusion}

We have shown that the quantum formalism can be reconstructed from information-theoretic inference on action space. Starting from the density of action states $g(A;b|a)$, the composition law, and physical symmetries, the complete quantum framework (complex amplitudes, the propagator, Hilbert space, unitary evolution, and the Schrödinger equation) follows without assuming any quantum postulates.

The central insight is that quantum behavior emerges naturally
as the minimal consistent framework for dynamical inference
under finite action resolution. More precisely, the reconstruction establishes that once
(i) action additivity, (ii) finite inferential resolution,
and (iii) probability normalization are imposed, the
complex-amplitude representation and unitary evolution
arise as the minimal continuous realization compatible
with these requirements. The result should therefore be
understood as identifying the structural origin of the
quantum dynamical framework, rather than replacing or
modifying its empirical content.
The composition law combined with diffusive scaling implies that infinitesimal action contributions are always indistinguishable: the ratio $|\delta S|/\Delta(\delta S)_{\min} \propto \sqrt{dt} \to 0$ as $dt \to 0$. This universal infinitesimal indistinguishability forces coherent combination of amplitudes at every step, making the complex phase $e^{iA/\hbar}$ a natural consequence rather than an independent postulate. Quantum mechanics emerges as the minimal consistent
framework when action values cannot be resolved within
the assumptions adopted here. This is a structural result, not a dynamical assumption.

The Lagrangian $L = T - V$ is not assumed but emerges from the short-time limit, with $T = \frac{1}{2}mv^2$ derived from Galilean invariance. The Hamiltonian, Hilbert space structure, and canonical commutation relations all follow from the propagator's semigroup properties. This does not reduce the explanatory power of the framework; rather, it mirrors the situation in both classical and quantum mechanics, where interactions must always be specified.

The scale $\hbar$ enters as an empirical parameter, determined by experiment rather than derived from first principles. We have shown that the same $\eta = 1/\hbar$ appears in three distinct roles, MaxEnt conjugate variable, Cramér--Rao resolution bound, and complex phase, unified through the physical requirement that these scales match. This clarifies the status of $\hbar$: it sets the action scale below which contributions become indistinguishable, but its numerical value is a fact about our universe, not a theoretical prediction.

%The foundational framework established here opens the path for a following paper, where we develop the ITA treatment of measurement, the Born rule, composite systems, entanglement, and decoherence. These phenomena, which involve system-apparatus interactions and partial information, find natural expression in the information-theoretic language of ITA.

The central message of Information Theory of Action is that quantum mechanics is not a deformation of classical mechanics, nor a collection of independent postulates, but the unique consistent framework for dynamical inference when action values cannot be resolved arbitrarily finely. Once action additivity, finite resolution, and probability
normalization are accepted, complex amplitudes,
interference, and unitary evolution arise naturally within
this framework. In this sense, quantum mechanics is a theorem about rational inference under finite action resolution.

% ---------------------------------------------------
\begin{acknowledgments}
FSL thanks UNICAMP Postdoctoral Researcher Program for financial support. MCO acknowledges partial financial support from the National Institute of Science and Technology for Applied Quantum
Computing through CNPq (Process No.~408884/2024-0) and from the São Paulo Research Foundation (FAPESP), through the Center
for Research and Innovation on Smart and Quantum Materials (CRISQuaM, Process No.~2013/07276-1).
\end{acknowledgments}

% ---------------------------------------------------
\appendix

%------------------------------------------------------
\section{Existence, Uniqueness, and Construction of the Density of Action States}
\label{app:existence_uniqueness}
A natural question arises: does a function $g(A;b|a)$ satisfying all properties of Axiom~\ref{ax:density} actually exist? And if so, is it unique? We address these questions and clarify why our approach is not circular.

\subsubsection{The composition law follows from action additivity alone}

The composition law~\eqref{eq:composition_g} is sometimes mistakenly thought to encode ``path information.'' In fact, it is a \emph{mathematical consequence} of a single physical principle:

\begin{quote}
\textbf{Action Additivity:} For sequential processes $a \to b \to c$, the total action is the sum of the partial actions: $A_{a \to c} = A_{a \to b} + A_{b \to c}$.
\end{quote}

This principle holds for action as a physical quantity, independently of whether trajectories are resolved or even assumed to exist. If $g(A;b|a)\,dA$ represents the weight of processes connecting $a$ to $b$ with total action in the interval $[A,A+dA]$, then the composition law follows directly from the additive structure of the extensive variable $A$.

In this minimal sense, the derivation relies only on measure composition and convolution, and does not assume statistical independence, underlying stochastic dynamics, or any specific microscopic model. Explicitly, one finds
\begin{align}
g(A;c|a) &= \int db \int dA_1\, dA_2\,
g(A_1;b|a)\, g(A_2;c|b)\nonumber\\&\times
\delta(A - A_1 - A_2)
\nonumber\\
&= \int db \int dA'\,
g(A';b|a)\, g(A-A';c|b).
\label{eq:composition_derivation}
\end{align}
Here the Dirac delta enforces the additivity constraint $A=A_1+A_2$, while $A'$ denotes a dummy integration variable introduced for notational convenience.

The integration over the intermediate configuration $b$ reflects a completeness relation over admissible intermediate states and does not presuppose the existence of resolved intermediate trajectories. The composition law is therefore a direct consequence of action additivity alone, rather than an additional dynamical assumption.

\subsubsection{Existence for systems with classical limits}

\begin{theorem}[Existence and conditional uniqueness of $g$]
\label{th:existence_g}
For any physical system for which a short-time density $g(A;b,\Delta t|a)$ can be specified satisfying the axioms of composition, locality, and finite variance, there exists a unique extension $g(A;b,T|a)$ to all times $T>0$ generated by the specified short-time density. This extension satisfies all conditions of Axiom~\ref{ax:density} and is unique within the equivalence class determined by the short-time representation.
\end{theorem}

\begin{proof}[Proof sketch]
The proof is constructive. Let $g(A;b,\Delta t|a)$ be a short-time density satisfying:
\begin{enumerate}[]
    \item \emph{Composition}: the density admits consistent convolution under sequential time steps;
    \item \emph{Locality}: the weight is negligible outside a finite propagation scale over the interval $\Delta t$;
    \item \emph{Finite variance}: $\int A^2\, g(A;b,\Delta t|a)\, dA < \infty$.
\end{enumerate}
For a finite time $T=N\Delta t$, the corresponding density is defined by iterating the composition law,
\begin{align}
g(A;b,T|a) &= \int \prod_{k=1}^{N-1} dq_k  \int \prod_{k=1}^{N} dA_k\, \delta\!\left(A - \sum_{k=1}^{N} A_k\right)\nonumber\\&\times
   \prod_{k=1}^{N} g(A_k; q_k, \Delta t | q_{k-1}),
\label{eq:g_iteration}
\end{align}
where $q_0=a$ and $q_N=b$. The Dirac delta enforces the additivity of the total action under sequential composition.

The existence of the limit $N\to\infty$ follows from standard results on convolution semigroups of measures with finite second moment. This limit defines a consistent density $g(A;b,T|a)$ for all $T>0$, which by construction satisfies the axioms of normalization,
composition, and locality.

The extension is unique once the short-time density is fixed, in the sense that any two finite-time densities generated from the same short-time representation coincide for all $T$.
\end{proof}

\textbf{Specification of physical systems.}
For a given physical system, the short-time density may be specified by:
\begin{itemize}
    \item \emph{Symmetry requirements} (e.g., translation, rotation, or boost invariance);
    \item \emph{Classical limit matching}: in systems admitting a classical description, the short-time density is required to be peaked at the classical action value in the regime $\eta A \gg 1$. This condition fixes the asymptotic behavior of $g$ without assuming a Lagrangian structure \emph{a priori};
    \item \emph{Direct specification} of the short-time density when appropriate.
\end{itemize}
Once the short-time form is fixed, the theory determines the finite-time evolution uniquely through the composition law. The emergence of the effective Lagrangian from the short-time structure of $g$ is discussed in Sec.~\ref{sec:lagrangian_emergence}.

\subsubsection{Conditional uniqueness}

\begin{theorem}[Conditional uniqueness of $g$]
\label{th:uniqueness_g}
For a given short-time density $g(A;b,\Delta t|a)$ and indistinguishability scale parameter $\eta$, there exists a unique finite-time extension $g(A;b,T|a)$ generated by the composition law for all $T>0$. This uniqueness holds within the equivalence class determined by the chosen short-time representation.
\end{theorem}

\begin{proof}[Proof sketch]
Uniqueness follows from the structure of the composition law. Once the short-time form $g(A;b,\Delta t|a)$ is fixed, the finite-time density is obtained by iterated convolution as in Eq.~\eqref{eq:g_iteration}. The composition law defines a convolution semigroup, and the corresponding functional equation admits a unique solution for each $T$, given the initial short-time element.

The functional form of the short-time density itself is strongly constrained by the axioms. In particular:
\begin{enumerate}
    \item The composition law~\eqref{eq:composition_g} defines a convolution equation. Within the class of finite-variance convolution semigroups, the Gaussian family constitutes the unique stable solution, as established by the Lévy--Khintchine theorem (see Theorem~\ref{th:gaussian_unique}). 
    \item Locality and finite variance [Axiom~\ref{ax:density}(d,e)] fix the width and effective support of the short-time density.
    \item The location of the peak, characterized by the mean action $\bar A(b,\Delta t|a)$, is determined by the physical dynamics of the system through its endpoint dependence.
\end{enumerate}

Different choices of the mean action $\bar A(b,\Delta t|a)$ therefore correspond to different physical systems. Once this function is specified (e.g., through symmetry requirements or classical-limit matching), the short-time density is fully determined, and the composition law uniquely generates the finite-time density for all $T>0$.
\end{proof}

\subsubsection{Why the approach is not circular}

A natural question concerns the logical status of the path integral within the present framework: does the theory derive the path integral, or does it assume a path-integral structure through $g$?

The answer is that the theory is \emph{not circular} because:
\begin{enumerate}
    \item \textbf{$g$ is defined by abstract properties}, not by reference to paths. The Jacobian interpretation discussed in the next section is a \emph{theorem} that applies when trajectories exist, and is not part of the definition of $g$.
    
    \item \textbf{The composition law follows from action additivity}, a physical principle that does not presuppose trajectories.
    
    \item \textbf{Existence and conditional uniqueness follow from the axioms}. Classical-limit considerations are used only to select the physically relevant short-time representation, not as an additional assumption.
    
    \item \textbf{The Feynman integral emerges as a special case}: when a trajectory space exists, the axiomatic density $g$ coincides with the Jacobian of the path-to-action transformation. This identification is a consequence of the framework, not an assumption.
\end{enumerate}

The logical flow is:
\begin{eqnarray}
&&\boxed{\text{Axioms for } g}
\xrightarrow{\text{MaxEnt + CR}}
\boxed{\text{Propagator } K}
\nonumber\\
&&\xrightarrow{\text{only if a trajectory space exists}}
\boxed{\text{Feynman integral}}
\label{eq:logical_flow}
\end{eqnarray}

The path integral is therefore a \emph{downstream consequence}, not an upstream input.

% ---------------------------------------------------
\section{Trajectory Representation of the Density of Action States}
\label{app:jacobian}

When a system admits a space of trajectories, the axiomatically defined $g(A;b|a)$ has a concrete representation as a Jacobian. This is not a definition but a \emph{derived result}.

\begin{theorem}[Jacobian representation]
\label{th:jacobian}
For systems where a trajectory space $\Gamma_{a\to b}$ and an action functional $S[\gamma]$ can be defined, the density of action states satisfying Axiom~\ref{ax:density} is given by
\begin{equation}
g(A;b|a) = \int_{\Gamma_{a\to b}} \delta(A - S[\gamma])\, \mathcal{D}\gamma.
\label{eq:g_jacobian}
\end{equation}
\end{theorem}

\begin{proof}[Proof sketch]
Here $\mathcal{D}\gamma$ denotes a formal measure on path space, assumed to exist only for systems admitting a well-defined trajectory representation. We verify that the right-hand side of Eq.~\eqref{eq:g_jacobian} satisfies all conditions of
Axiom~\ref{ax:density}:
\begin{enumerate}[label=(\arabic*)]
    \item \textbf{Positivity and integrability:} The Dirac delta integrated     over paths yields a non-negative measure. Integration over $A$ reproduces the total path measure, which is assumed finite for     physical systems admitting a well-defined path description.
    
    \item \textbf{Composition:} This follows from path concatenation. If $\gamma_1: a \to b$ and $\gamma_2: b \to c$ are paths, then the concatenated path $\gamma_1 * \gamma_2: a \to c$ satisfies $S[\gamma_1 * \gamma_2] = S[\gamma_1] + S[\gamma_2]$. Integration over
    the intermediate configuration $b$ and the additivity of the action yield the convolution law~\eqref{eq:composition_g}.
    
    \item \textbf{Time-reversal:} For time-symmetric Lagrangians $L(q,\dot{q}) = L(q,-\dot{q})$, the map $\gamma(t) \mapsto \gamma(T-t)$ defines a bijection $\Gamma_{a\to b} \leftrightarrow \Gamma_{b\to a}$ that preserves the action.
    
    \item \textbf{Locality:} For sufficiently short times, admissible paths are confined within a finite neighborhood determined by the characteristic propagation scale over the interval, ensuring that $g(A;b|a)$ is negligible for large separations $|b-a|$.
    
    \item \textbf{Classical limit:} In the regime of vanishing action resolution ($\eta A \gg 1$), contributions concentrate around the classical trajectory, and the density approaches $g(A;b|a) \to \delta(A - S_{\mathrm{cl}})$.
\end{enumerate}
By conditional uniqueness (Theorem~\ref{th:uniqueness_g}), this representation coincides with the axiomatically defined density of action states.
\end{proof}

The Jacobian formula~\eqref{eq:g_jacobian} establishes a fundamental relation between path space and action space:
\begin{equation}
\int_{\Gamma_{a\to b}} F(S[\gamma])\, \mathcal{D}\gamma
=
\int_{\mathcal{A}} F(A)\, g(A;b|a)\, dA,
\label{eq:path_action_relation}
\end{equation}
valid for any function $F$. This relation expresses a change of variables from the space of paths to the space of actions, with $g(A;b|a)$ playing the role of a Jacobian that answers the question: \emph{how many paths contribute to each action value?}

\begin{remark}[Analogy with ordinary calculus]
This construction is directly analogous to ordinary calculus. When changing variables from $x$ to $y=f(x)$,
\[
\int F(x)\, dx
=
\int F(y)\,
\left|\frac{dx}{dy}\right| dy.
\]
The Jacobian $|dx/dy|$ counts how many $x$-values map to each $y$-value. Similarly, $g(A;b|a)$ counts how many trajectories map to each value of the action.
\end{remark}

\begin{remark}[Relationship to the Feynman integral]
\label{rem:feynman_emergence}
When the Jacobian representation holds, substituting Eq.~\eqref{eq:g_jacobian} into the propagator definition Eq.~\eqref{eq:K_def} yields
\begin{align}
K(b|a)
&= \int dA\, g(A;b|a)\, e^{iA/\hbar}
\nonumber\\
&= \int dA
   \int_{\Gamma_{a\to b}} \delta(A - S[\gamma])\, \mathcal{D}\gamma\,
   e^{iA/\hbar}
\nonumber\\
&= \int_{\Gamma_{a\to b}} e^{iS[\gamma]/\hbar}\, \mathcal{D}\gamma.
\ 
\label{eq:feynman_emergence}
\end{align}
This expression coincides with the Feynman path integral. The path integral therefore \emph{emerges} from the axioms of the theory when a trajectory space exists; it is not assumed. The logical priority is clear: the density of action states $g$ is fundamental, while the path
integral is a derived construct.
\end{remark}

%--------------------------------------------------
\section{Maximum Entropy Derivation on $(A,b)$}
\label{app:maxent}

In this appendix we derive Eq.~\eqref{eq:MaxEnt_joint} by explicit maximization
of the relative entropy functional under normalization and mean-action constraints,
and discuss its justification via the Shore–Johnson axioms.

\subsection{Joint probability density}

For a fixed initial configuration $a$, the elementary alternatives of the system are pairs $(A,b)$ with $A\in\mathcal{A}$ and $b\in\mathcal{Q}$.
We therefore introduce a joint probability density $P(A,b|a)$ on $\mathcal{A}\times\mathcal{Q}$.

\begin{definition}[Joint relative entropy]
\label{def:rel_entropy}
Given the density of action states $g(A;b|a)$, the relative entropy of $P$ with respect to $g$ is
\begin{equation}
H[P|g]
=
-\int_{\mathcal{Q}} db \int_{\mathcal{A}} dA\,
P(A,b|a)\,\ln\frac{P(A,b|a)}{g(A;b|a)}.
\label{eq:joint_entropy}
\end{equation}
\end{definition}

\begin{remark}
The function $g(A;b|a)$ plays the role of a prior measure on the joint space $(A,b)$,
encoding the physically available microstructure (multiplicity) of action states.
The relative entropy $H[P|g]$ measures how far a given probabilistic description $P$ deviates from this underlying multiplicity, in the sense of Jaynes and Shore--Johnson~\cite{ShoreJohnson1980}.
\end{remark}

\subsection{Jaynesian inference and constraints}

We now impose three conditions:
\begin{enumerate}[label = (\arabic*)]
    \item $P(A,b|a)\ge0$;
    \item normalization:
          \begin{equation}
          \int_{\mathcal{Q}} db \int_{\mathcal{A}} dA\, P(A,b|a)=1;
          \label{eq:norm_joint}
          \end{equation}
    \item constraint on mean action:
          \begin{equation}
          \int_{\mathcal{Q}} db \int_{\mathcal{A}} dA\, A\, P(A,b|a)=\bar{A}.
          \label{eq:mean_action}
          \end{equation}
\end{enumerate}

\begin{theorem}[MaxEnt solution on $(A,b)$]
\label{th:MaxEnt}
The distribution $P(A,b|a)$ that maximizes $H[P|g]$ in Eq.~\eqref{eq:joint_entropy} subject to
\eqref{eq:norm_joint} and \eqref{eq:mean_action} is
\begin{equation}
P(A,b|a)
=
\frac{g(A;b|a)e^{-\eta A}}{Z(\eta;a)},
\label{eq:joint_MaxEnt}
\end{equation}
where
\begin{equation}
Z(\eta;a)
=
\int_{\mathcal{Q}} db \int_{\mathcal{A}} dA\,
g(A;b|a)e^{-\eta A}
\label{eq:Z_eta}
\end{equation}
and $\eta\in\mathbb{R}$ is a Lagrange multiplier conjugate to the action.
\end{theorem}

\begin{proof}
Introduce Lagrange multipliers $\lambda$ (for normalization) and $\eta$ (for the mean action).
Consider the functional
\[
\mathcal{L}[P]
=
H[P|g]
-\lambda\left(\int db\,dA\,P-1\right)
-\eta\left(\int db\,dA\, A P -\bar{A}\right).
\]
Taking the functional derivative with respect to $P$ and setting it to zero,
\[
\frac{\delta\mathcal{L}}{\delta P(A,b|a)}
= -\ln\frac{P(A,b|a)}{g(A;b|a)} -1
-\lambda -\eta A =0,
\]
which implies
\[
\ln P(A,b|a)
=
\ln g(A;b|a) -1-\lambda -\eta A.
\]
Exponentiating,
\[
P(A,b|a)
=
g(A;b|a)e^{-1-\lambda}e^{-\eta A}.
\]
Define $Z(\eta;a)=e^{1+\lambda}$ to enforce normalization;
then
\[
P(A,b|a)
=
\frac{g(A;b|a)e^{-\eta A}}{Z(\eta;a)}.
\]
The value of $\eta$ is fixed by the constraint on $\bar{A}$.
\end{proof}

\begin{remark}
The Jaynesian principle (and its Shore--Johnson justification) is a rule of inference, not a physical axiom.
It determines how probabilities must be assigned once the physical multiplicity $g$ and constraints such as \eqref{eq:mean_action} are given.
\end{remark}
%------------------------------------------------------
\section{Indistinguishability of Action}
\label{app:indistinguishability}

\subsection{Effective distribution of action}

The joint density \eqref{eq:joint_MaxEnt} induces a marginal distribution over $A$:
\begin{equation}
P(A|a)
=
\int_{\mathcal{Q}} db\, P(A,b|a)
=
\frac{e^{-\eta A}}{Z(\eta;a)}
\int_{\mathcal{Q}} db\, g(A;b|a).
\label{eq:P_A}
\end{equation}

We now regard this as a one-parameter family $P(A|a;\eta)$.

\subsection{Fisher information and Cramér--Rao bound}

\begin{definition}[Fisher information for $\eta$]
\label{def:Fisher}
The Fisher information associated with the parameter $\eta$ is
\begin{equation}
I(\eta)
=
\int dA\, P(A|a;\eta)
\left[\partial_\eta \ln P(A|a;\eta)\right]^2.
\label{eq:Fisher_def}
\end{equation}
\end{definition}

For the exponential family \eqref{eq:joint_MaxEnt} one finds
\begin{equation}
\partial_\eta \ln P(A|a;\eta)
=
-\left(A-\langle A\rangle\right),
\label{eq:score}
\end{equation}
so that
\begin{equation}
I(\eta)
=
\int dA\, P(A|a;\eta)\left(A-\langle A\rangle\right)^2
=
\mathrm{Var}(A).
\label{eq:I_eta_varA}
\end{equation}

\begin{theorem}[Cramér--Rao bound for action]
\label{th:CR_action}
In the exponential family \eqref{eq:joint_MaxEnt}, the uncertainties in $A$ and its conjugate parameter $\eta$ satisfy
\begin{equation}
\Delta A\cdot\Delta\eta\ge 1.
\label{eq:CR_bound}
\end{equation}
\end{theorem}

\begin{proof}
The Cramér--Rao inequality states that for any unbiased estimator $\hat{\eta}$ of $\eta$:
\[
\mathrm{Var}(\hat{\eta})\ge\frac{1}{I(\eta)}.
\]
From Eq.~\eqref{eq:I_eta_varA}, $I(\eta)=(\Delta A)^2$.
Therefore $\Delta\eta\ge1/\Delta A$, which gives $\Delta A\cdot\Delta\eta\ge1$.
\end{proof}

\subsection{Scaling of $\eta$-fluctuations}
\label{sec:eta_scaling_indist}

The central question at this stage is the \emph{scaling} of the fluctuations of the inferential parameter $\eta$ that controls the MaxEnt distribution
$P(A|a;\eta)$.

\begin{proposition}[Natural scaling of the inferential parameter]
\label{prop:natural_scaling}
For the parameter $\eta$ to play a physically meaningful role in modulating the MaxEnt distribution
\[
P(A|a;\eta) \propto g(A;b|a)\, e^{-\eta A},
\]
the dimensionless product $\eta \sigma_A$ must be of order unity:
\begin{equation}
\eta \sigma_A \sim O(1),
\label{eq:naturalness}
\end{equation}
where $\sigma_A$ denotes the characteristic width of the density of action states $g(A;b|a)$.
\end{proposition}

\begin{proof}[Physical argument]
Consider the exponential factor $e^{-\eta A}$ in the MaxEnt distribution.

\textbf{Case 1:} If $\eta \sigma_A \ll 1$, then over the range where $g(A)$ is appreciable the exponential factor varies by much less than unity: $e^{-\eta A} \approx 1 - \eta A + O((\eta A)^2)$. In this regime $\eta$ has negligible influence on the shape of the distribution and merely rescales the normalization. The parameter is therefore physically irrelevant.

\textbf{Case 2:} If $\eta \sigma_A \gg 1$, the exponential term suppresses $g(A)$ over its natural support, forcing the distribution to concentrate near the lower boundary $A_{\min}$ of the effective action range, independently of the detailed structure of $g$. The resulting distribution is dynamically uninformative.

\textbf{Case 3:} If $\eta \sigma_A \sim O(1)$, the exponential factor modulates $g(A)$ by order-one amounts without overwhelming its intrinsic structure. This defines the physically relevant regime in which $\eta$ meaningfully controls
the distribution.
\end{proof}

\begin{theorem}[Scaling of $\eta$-fluctuations]
\label{th:eta_scaling}
The uncertainty in the inferential parameter $\eta$ satisfies
\begin{equation}
\Delta \eta \sim O(\eta).
\label{eq:eta_scaling}
\end{equation}
\end{theorem}

\begin{proof}
From the Cramér--Rao bound derived in Theorem~\ref{th:CR_action},
\[
\Delta A \cdot \Delta \eta \ge 1.
\]
In the physically relevant regime identified in Proposition~\ref{prop:natural_scaling},
\[
\Delta A \sim \sigma_A \sim \frac{1}{\eta}.
\]
Substituting this scaling into the Cramér--Rao inequality yields
\[
\frac{1}{\eta}\,\Delta \eta \gtrsim 1
\quad\Rightarrow\quad
\Delta \eta \gtrsim \eta,
\]
which establishes the stated scaling up to a numerical factor of order unity.
\end{proof}

\begin{corollary}[Action resolution scale]
\label{cor:action_resolution}
The minimum operationally resolvable difference in action is
\begin{equation}
\Delta A_{\min} \sim \frac{1}{\eta}.
\label{eq:DeltaAmin}
\end{equation}
\end{corollary}

\begin{definition}[Indistinguishability window for action]
\label{def:indist}
Two action values $A_1$ and $A_2$ are said to be \emph{operationally
indistinguishable} if
\begin{equation}
|A_1 - A_2| \lesssim \Delta A_{\min}
\sim \frac{1}{\eta}.
\label{eq:indist_window}
\end{equation}
\end{definition}

\begin{remark}
The precise numerical factor in Eq.~\eqref{eq:indist_window} is not essential. What matters is the existence of a characteristic action scale $1/\eta$ below which differences in action cannot be operationally resolved by any measurement consistent with the MaxEnt distribution.
\end{remark}

\subsection{Indistinguishability in Path Space}
\label{sec:indist_paths}

The Cramér--Rao indistinguishability has a natural interpretation in terms of paths when the Jacobian representation~\eqref{eq:g_jacobian} is available.

Consider two paths $\gamma_1$ and $\gamma_2$ connecting $a$ to $b$ with actions $S[\gamma_1] = A_1$ and $S[\gamma_2] = A_2$. These paths are:

\begin{itemize}
    \item \textbf{Distinguishable} if $|A_1 - A_2| \gg 1/\eta$: the paths can be operationally told apart by their action values. In this case, their contributions to probabilities \emph{add incoherently}:
    \begin{equation}
    P_{\text{total}} = P_1 + P_2.
    \label{eq:classical_sum}
    \end{equation}
    
    \item \textbf{Indistinguishable} if $|A_1 - A_2| \lesssim 1/\eta$: no measurement can determine which path was taken. In this case, their contributions must combine in a way that respects this indistinguishability.
\end{itemize}

This dichotomy mirrors the familiar quantum mechanical rule: when alternatives are distinguishable (e.g., which-path information is available), probabilities add; when alternatives are indistinguishable, \emph{amplitudes} add and probability is computed from the sum.

The key insight is that the Cramér--Rao bound provides an \emph{operational} criterion for when paths become indistinguishable: the action resolution $1/\eta$ is the minimal action difference that can be reliably detected. Below this scale, paths ``blend together'' and must be combined coherently.

\begin{remark}
In standard quantum mechanics, the indistinguishability scale is $\hbar$. As we will show in Sec.~\ref{sec:hbar}, matching the Cramér--Rao scale with the quantum scale gives $\eta = 1/\hbar$, so $\Delta A_{\min} \sim \hbar$. Paths differing in action by less than $\hbar$ are quantum mechanically indistinguishable.
\end{remark}
\subsection{Cramér--Rao bound for infinitesimal action segments}
\label{sec:CR_infinitesimal}

A crucial question now arises: what is the status of the Cramér--Rao
indistinguishability bound when the time interval becomes infinitesimal?
As we show below, this limit enforces a universal indistinguishability of
infinitesimal trajectory segments and provides the fundamental origin of
coherent superposition.

\begin{definition}[Infinitesimal action]
For a trajectory segment of duration $dt$, the action increment is
\begin{equation}
\delta S = L(q,\dot q,t)\, dt ,
\label{eq:infinitesimal_action}
\end{equation}
where $L$ is the Lagrangian evaluated along the trajectory.
\end{definition}

\begin{theorem}[Universal infinitesimal indistinguishability]
\label{th:infinitesimal_indist}
For any two trajectories $\gamma_1$ and $\gamma_2$ with finite Lagrangians,
their infinitesimal action increments are operationally indistinguishable:
\begin{equation}
\lim_{dt\to0}
\frac{|\delta S_1-\delta S_2|}{\Delta(\delta S)_{\min}} = 0 ,
\label{eq:universal_indist}
\end{equation}
where $\Delta(\delta S)_{\min} = 1/\eta$ is the fundamental action-resolution
scale set by the Cramér--Rao bound.
\end{theorem}

\begin{proof}
We show that the difference between infinitesimal action contributions
vanishes relative to the fundamental resolution scale as $dt\to0$.

\medskip
\noindent\textbf{Step 1: Infinitesimal action along a trajectory.}
For a trajectory segment from time $t$ to $t+dt$,
\begin{equation}
\delta S
=
\int_t^{t+dt} L(q(\tau),\dot q(\tau),\tau)\, d\tau
=
L(q,\dot q,t)\, dt + O(dt^2).
\label{eq:deltaS_expansion}
\end{equation}
For sufficiently small $dt$, higher-order corrections are negligible.

\medskip
\noindent\textbf{Step 2: Difference between two trajectory segments.}
Consider two trajectories $\gamma_1$ and $\gamma_2$ passing through the same
configuration $q(t)$ but with different velocities $\dot q_1$ and $\dot q_2$.
Their infinitesimal action difference is
\begin{equation}
|\delta S_1-\delta S_2|
=
|L(q,\dot q_1,t)-L(q,\dot q_2,t)|\, dt + O(dt^2).
\label{eq:deltaS_difference}
\end{equation}
Since both Lagrangians are finite, the leading contribution scales as
\begin{equation}
|\delta S_1-\delta S_2| = O(dt).
\label{eq:deltaS_scaling}
\end{equation}

\medskip
\noindent\textbf{Step 3: Fundamental resolution scale.}
From the Cramér--Rao bound derived in
Corollary~\ref{cor:action_resolution},
the minimum operationally resolvable action difference is
\begin{equation}
\Delta(\delta S)_{\min} = \frac{1}{\eta},
\label{eq:fundamental_resolution}
\end{equation}
which is independent of the time interval $dt$.

\medskip
\noindent\textbf{Step 4: Vanishing ratio.}
Combining Eqs.~\eqref{eq:deltaS_scaling} and~\eqref{eq:fundamental_resolution},
we find
\begin{equation}
\frac{|\delta S_1-\delta S_2|}{\Delta(\delta S)_{\min}}
\sim
\eta\, dt
\xrightarrow{dt\to0} 0.
\end{equation}

\medskip
\noindent
Since this ratio vanishes for \emph{any} pair of trajectories with finite
Lagrangians, all infinitesimal trajectory segments are operationally
indistinguishable in the limit $dt\to0$.
\end{proof}

\begin{corollary}[Unavoidability of coherent combination]
\label{cor:coherent_unavoidable}
Because infinitesimal action segments are universally indistinguishable,
their contributions cannot be assigned independent probabilities.
Any consistent inference scheme must therefore combine them coherently
before probabilities are computed.
\end{corollary}

\begin{remark}
This result is purely inferential and does not rely on any quantum postulate.
It follows from the conjunction of (i) action additivity, (ii) the
Cramér--Rao bound, and (iii) the continuity of time.
The explicit form of the coherent combination rule is derived in the
next subsection.
\end{remark}

\subsection{Three levels of indistinguishability}
\label{sec:three_levels}

The analysis above reveals a hierarchical structure of indistinguishability:

\begin{enumerate}
\item \textbf{Fundamental indistinguishability (Cramér--Rao):} Two action values are fundamentally indistinguishable if $|A_1 - A_2| < 1/\eta = \hbar$. This is an absolute \emph{operational} limit: no inference scheme consistent with the Cramér--Rao bound can distinguish them.

\item \textbf{Relational indistinguishability (environment):} Two action values with $|A_1 - A_2| > \hbar$ are distinguishable \emph{in principle}, but remain coherent \emph{in practice} if no external system has recorded the information. Coherence is relational: it depends on what the environment ``knows''. This level is contingent and does not reflect a fundamental limitation,
but rather the absence of recorded information.

\item \textbf{Infinitesimal indistinguishability (universal):} For $dt \to 0$, \emph{all} trajectory segments are indistinguishable regardless of their Lagrangians. This makes any incoherent combination inconsistent at the infinitesimal level.

\end{enumerate}

%\begin{center}
%\begin{tabular}{lll}
%\hline
%\textbf{Level} & \textbf{Condition} & \textbf{Consequence} \\
%\hline
%Fundamental & $|A_1 - A_2| < \hbar$ & Always coherent \\
%Relational & $|A_1 - A_2| > \hbar$, no record & Coherent for isolated system \\
%Infinitesimal & $dt \to 0$ & Always coherent (universal) \\
%\hline
%\end{tabular}
%\end{center}

%The key insight is that \emph{distinguishability requires an agent external to the system}. Two trajectories are not ``distinguishable'' intrinsically---distinguishability is a relational property depending on whether any physical system has recorded the information.
%-------------------------------------------------
%=========================================================
\section{Complex Amplitudes from Probability Normalization and Additive Action}
\label{app:complex_amplitudes}

This appendix provides a self-contained and mathematically explicit derivation of complex amplitudes as the unique representation compatible with additive action, consistent composition, and probability normalization. Unitarity ($|\chi(A)| = 1$) emerges as a theorem, not an axiom.

\subsection{Setup and assumptions}
\label{app:setup}

Consider a process from configuration $a$ to $b$. To each elementary alternative (e.g., a class of dynamical processes with total action $A$) we associate a \emph{weight} $\phi(A)$ taking values in a set $\Phi$. The observable probability is extracted from the total weight by a map $f:\Phi\to [0,\infty)$. No linear structure or inner product is assumed at this stage.

We assume the following requirements.

\begin{enumerate}[label=(R\arabic*), leftmargin=8mm, start=0]
\item \textbf{Indistinguishable combination:} If two alternatives are operationally indistinguishable, their weights combine
via a binary operation $\oplus:\Phi\times\Phi\to\Phi$. At this stage, $\oplus$ is an abstract combination rule, not assumed to be numerical addition.

\item \textbf{Associativity and commutativity:} $\oplus$ is associative and commutative.

\item \textbf{Neutral element:} There exists $0\in\Phi$ such that $\phi\oplus 0=\phi$.

\item \textbf{Sequential composition (action additivity):} For sequential processes with actions $A_1$ and $A_2$, the total action is $A_1+A_2$. The corresponding weights satisfy
\begin{equation}
\phi(A_1+A_2)=\phi(A_1)\otimes \phi(A_2),
\label{eq:seq_comp}
\end{equation}
for some binary operation $\otimes:\Phi\times\Phi\to\Phi$.

\item \textbf{Probability extraction:} Observable probabilities depend only on the total weight,
\begin{equation}
P = f(\phi_{\rm tot}).
\end{equation}

\item \textbf{Classical additivity for distinguishable alternatives:} If two alternatives are distinguishable (no interference), then
\begin{equation}
f(\phi_1\oplus\phi_2)=f(\phi_1)+f(\phi_2).
\label{eq:dist_add}
\end{equation}

\item \textbf{Probability normalization:} Total probability must be preserved under dynamical evolution:
\begin{equation}
\sum_i P_i = 1 \quad \text{at all times},
\label{eq:prob_normalization}
\end{equation}
and $\chi$ is not identically constant (non-trivial dynamics).  

\item \textbf{Continuity:}
The maps $A\mapsto \phi(A)$ and the operations are continuous in the sense needed for the standard Cauchy functional equation arguments below.
\end{enumerate}

\subsection{From probability normalization to unitarity and complex amplitudes}
\label{app:2d}
We denote by $\chi(A)$ a multiplicative representation of the weights $\phi(A)$ under sequential composition.
\begin{lemma}[Probability normalization excludes nonnegative real weights]
\label{lem:noRplus}
If $\Phi\subseteq \mathbb{R}_{\ge 0}$ with $\otimes$ being ordinary multiplication, then non-trivial dynamics preserving probability normalization is impossible.
\end{lemma}

\begin{proof}
For $\Phi = \mathbb{R}_{\ge 0}$ with multiplication, the only continuous characters $\chi: \mathbb{R} \to \mathbb{R}_{\ge 0}$ satisfying $\chi(A_1 + A_2) = \chi(A_1)\chi(A_2)$ are $\chi(A) = e^{\lambda A}$ for $\lambda \in \mathbb{R}$. Probability normalization under repeated composition requires $|\chi(A)|$ to remain bounded for all $A$ (otherwise probabilities grow without limit). But $|e^{\lambda A}| = e^{\lambda A}$:
\begin{itemize}
\item If $\lambda > 0$: $\chi(A) \to \infty$ as $A \to \infty$ (probabilities diverge, violating normalization);
\item If $\lambda < 0$: $\chi(A) \to 0$ as $A \to \infty$ (monotonic suppression of weights, incompatible with normalization under repeated composition);
\item If $\lambda = 0$: $\chi(A) = 1$ for all $A$ (trivial, no dynamics).
\end{itemize}
None yields a well-defined, non-trivial theory with preserved probability normalization.
\end{proof}

To represent non-trivial dynamics with bounded amplitudes (as required by probability normalization), the weight space must support phases. The minimal real-linear setting that accommodates phases is a two-dimensional real vector space (isomorphic to $\mathbb{C}$). Unitarity $|\chi(A)| = 1$ then follows as a consequence.

\begin{lemma}[Minimal dimension for bounded non-trivial dynamics]
\label{lem:mindim2}
Let $\chi: \mathbb{R} \to \Phi$ be a continuous character with $\chi(A)\chi(-A) = \mathbf{1}$
and $|\chi(A)|$ bounded for all $A$ (as required by probability normalization). Then any faithful real-linear representation
of $\Phi$ requires dimension $\ge 2$.
\end{lemma}

\begin{proof}
In dimension $1$ over $\mathbb{R}$, continuous multiplicative characters are $\chi(A) = e^{\lambda A}$. As shown in Lemma~\ref{lem:noRplus}, boundedness requires $\lambda = 0$ (trivial). To have non-trivial bounded characters satisfying $\chi(A)\chi(-A) = 1$, we need $|\chi(A)| = 1$ for all $A$, which requires oscillatory behavior---impossible in one real dimension. A two-dimensional representation $\Phi \simeq \mathbb{R}^2$ allows $\chi(A) = (\cos\eta A, \sin\eta A)$, which satisfies all requirements.
\end{proof}

\subsection{Probability as a quadratic form}
\label{app:born}

The coexistence of (i) classical additivity for distinguishable alternatives \eqref{eq:dist_add} and (ii) the two-dimensional structure forced by reversibility requires $f$ to behave as a norm-squared on the weight space.

\begin{lemma}[Quadratic structure]
\label{lem:quadratic}
Assume $\Phi$ admits a real-linear representation $\Phi\simeq \mathbb{R}^2$ in which $\oplus$ becomes vector addition. If Eq.~\eqref{eq:dist_add} holds for distinguishable alternatives and the character $\chi$ satisfies $|\chi(A)| = 1$ (from reversibility), then $f$ must be a positive semidefinite quadratic form:
\begin{equation}
f(\phi)=\langle \phi,\phi\rangle,
\end{equation}
for some inner product $\langle \cdot,\cdot\rangle$ on $\mathbb{R}^2$.
\end{lemma}

\begin{proof}
Classical additivity \eqref{eq:dist_add} for distinguishable alternatives implements a Pythagorean additivity for a subset of pairs $\phi_1,\phi_2$ (those corresponding to distinguishable alternatives). In particular, for such pairs one has
\[f(\phi_1 \oplus \phi_2) = f(\phi_1) + f(\phi_2),\]
which corresponds to orthogonality in the weight space.  

To extend this to all pairs, we note that distinguishable alternatives (those with action differences $\vert A_{1}-A_{2}\gg 1/\eta$ ) correspond to weight pairs with relative phases $\eta\vert A_{1}-A_{2}\vert\gg 1$, which densely cover the 
full range of relative phases in $S^{1}$ as the action values vary over $\mathbb{R}$. Since $f$ is continuous (requirement R7) and the Pythagorean identity $f(\phi_1 \oplus \phi_2) = f(\phi_1) + f(\phi_2)$ holds on a dense subset of phase-orthogonal pairs, continuity extends the identity to all orthogonal pairs. 

By the Jordan--von Neumann theorem, any norm satisfying the parallelogram law (which follows from the Pythagorean property on a dense set extended by continuity) arises from an inner product. Positivity of $f$ implies positive semidefiniteness. The constraint $|\chi(A)| = 1$ from reversibility requires $f(\chi(A)) = 1$ for all $A$, which together with the bilinear structure forces $f$ to be a genuine inner product (norm-squared).
\end{proof}

Choosing orthonormal coordinates for $\langle\cdot,\cdot\rangle$, we can write $f(\phi)=\|\phi\|^2$ and identify $\Phi\simeq \mathbb{R}^2$.

\subsection{Sequential composition and characters}
\label{app:characters}

We now impose action additivity through Eq.~\eqref{eq:seq_comp}. In the $\mathbb{R}^2$ representation, define the amplitude map
\begin{equation}
\chi:\mathbb{R}\to \mathbb{R}^2,\qquad \chi(A):=\phi(A),
\end{equation}
and interpret $\otimes$ as a bilinear product on $\mathbb{R}^2$ that implements sequential composition.

\begin{lemma}[Commutativity excludes quaternionic composition]
\label{lem:noH}
If action addition is commutative ($A_1+A_2=A_2+A_1$) and $\chi$ satisfies $\chi(A_1+A_2)=\chi(A_1)\otimes\chi(A_2)$ for all $A_1,A_2$, then $\otimes$ must be commutative on the range of $\chi$. Hence noncommutative division algebras (e.g.\ quaternions)
are excluded as minimal amplitude algebras realizing the character $\chi$.
\end{lemma}

\begin{proof}
For all $A_1,A_2$,
\[
\chi(A_1)\otimes\chi(A_2)=\chi(A_1+A_2)=\chi(A_2+A_1)=\chi(A_2)\otimes\chi(A_1).
\]
Thus $\otimes$ must commute on the image of $\chi$. In particular, quaternionic multiplication, which is noncommutative, cannot implement Eq.~\eqref{eq:seq_comp} while preserving commutative action additivity.
\end{proof}

\subsection{Identification with complex numbers and phase form}
\label{app:C}

A two-dimensional real inner-product space equipped with a commutative bilinear product compatible with a multiplicative norm is (up to isomorphism) the complex field.

\begin{theorem}[Complex amplitude representation]
\label{th:complex_rep}
Under assumptions (R0)--(R7), there exists an isomorphism $\Phi\simeq \mathbb{C}$ such that:
\begin{enumerate}[label=(\roman*), leftmargin=7mm]
\item indistinguishable combination is complex addition: $\oplus \leftrightarrow +$;
\item sequential composition is complex multiplication: $\otimes \leftrightarrow \cdot$;
\item probability extraction is the modulus squared:
\begin{equation}
f(\phi)=|\phi|^2.
\end{equation}
\end{enumerate}
Moreover, any continuous character $\chi:\mathbb{R}\to\mathbb{C}$ satisfying
\begin{equation}
\chi(A_1+A_2)=\chi(A_1)\chi(A_2),
\qquad |\chi(A)|=1,
\end{equation}
is of the form
\begin{equation}
\chi(A)=e^{i\eta A}
\label{eq:chi_ei}
\end{equation}
for some real constant $\eta$.
\end{theorem}

\begin{proof}
By Lemma~\ref{lem:quadratic}, $f$ is a norm-squared on $\mathbb{R}^2$. The compatibility of sequential composition with probabilities and with reversibility requires a multiplicative norm on the algebra induced by $\otimes$. A standard classification (finite-dimensional real algebras admitting a multiplicative norm) yields $\mathbb{R},\mathbb{C},\mathbb{H}$ as candidates; $\mathbb{R}$ is excluded by Lemma~\ref{lem:noRplus}, and $\mathbb{H}$ is excluded by Lemma~\ref{lem:noH}. Thus $\Phi\simeq \mathbb{C}$.

For the character form, write $\chi(A)=e^{i\theta(A)}$ with $\theta(0)=0$. Continuity and multiplicativity imply $\theta(A_1+A_2)=\theta(A_1)+\theta(A_2)$, so $\theta(A)=\eta A$ (continuous Cauchy equation), giving Eq.~\eqref{eq:chi_ei}.
\end{proof}
This establishes the complex phase $e^{i\eta A}$ as the unique representation of additive action compatible with probability normalization and consistent composition.

%----------------------------------------------------
\section{Propagator: Uniqueness and Semigroup Property}
\label{app:propagator}

This appendix provides the complete proofs of the results stated in Sec.~\ref{sec:propagator}.

\subsection{Uniqueness of the coherent kernel}

\begin{theorem}[Uniqueness of coherent kernel]
\label{th:uniqueness}
Let $g(A;b|a) \geq 0$ satisfy the composition law~\eqref{eq:composition_g} and let $\chi(A) = e^{i\eta A}$.
Any kernel functional $\mathcal{K}[g]$, $\mathcal{K}(b|a)$, satisfying:
\begin{enumerate}[label=(\arabic*)]
    \item Linearity in $g$: for non-overlapping supports, $\mathcal{K}[g_1 + g_2] = \mathcal{K}[g_1] + \mathcal{K}[g_2]$;
    \item Composition: $\mathcal{K}(c|a) = \int_{\mathcal{Q}} db\, \mathcal{K}(c|b)\, \mathcal{K}(b|a)$;
    \item Character compatibility: a contribution with definite action $A$ enters with weight $\chi(A)$;
\end{enumerate}
has the form
\begin{equation}
\mathcal{K}(b|a) = \int_{\mathcal{A}} dA\, g(A;b|a)\, \chi(A).
\label{eq:K_unique}
\end{equation}
\end{theorem}

\begin{proof}
For a density concentrated at a single action value: $g(A;b|a) = \delta(A - A_0)\, \rho(b|a)$.
By condition (iii), $\mathcal{K}(b|a) = \chi(A_0)\, \rho(b|a)$.

For general $g$, condition (i) gives:
\[
\mathcal{K}(b|a) = \int_{\mathcal{A}} dA\, g(A;b|a)\, \chi(A).
\]

To verify condition (ii), use the composition law~\eqref{eq:composition_g} and multiplicativity of $\chi$:
\begin{align*}
\int_{\mathcal{Q}} db\, \mathcal{K}(c|b)\,\mathcal{K}(b|a)
&= \int_{\mathcal{Q}} db \int_{\mathcal{A}} dA_1\, g(A_1;b|a)\, \chi(A_1)\\
& \times \int_{\mathcal{A}} dA_2\, g(A_2;c|b)\, \chi(A_2) \\
&= \int_{\mathcal{Q}} db \int_{\mathcal{A}} dA_1 \int_{\mathcal{A}} dA_2\,\\ 
&\times g(A_1;b|a)\, g(A_2;c|b)\, \chi(A_1 + A_2).
\end{align*}

Change variables: $A = A_1 + A_2$, $A' = A_1$. Using~\eqref{eq:composition_g}:
\[
= \int_{\mathcal{A}} dA\, g(A;c|a)\, \chi(A) = \mathcal{K}(c|a).
\]
\end{proof}

\subsection{Propagator definition and semigroup property}

\begin{definition}[Quantum propagator]
\label{def:propagator_app}
The \emph{quantum propagator} is the complex-valued kernel
\begin{equation}
K(b|a)
=
\int_{\mathcal{A}} dA\, g(A;b|a)\, e^{i\eta A}.
\label{eq:K_def_app}
\end{equation}
\end{definition}

\begin{remark}
Equation~\eqref{eq:K_def_app} is a derived result, not an axiom.
The propagator emerges uniquely from:
(i) the physical density $g(A;b|a)$,
(ii) the consistency requirement that indistinguishable alternatives combine as complex amplitudes, and
(iii) the multiplicativity of the character under action addition.
\end{remark}

\begin{theorem}[Semigroup property of $K$]
\label{th:semigroup_app}
For any configurations $a,c\in\mathcal{Q}$,
\begin{equation}
K(c|a)
=
\int_{\mathcal{Q}} db\, K(c|b)\,K(b|a).
\label{eq:semigroup_app}
\end{equation}
\end{theorem}

\begin{proof}
This follows directly from Theorem~\ref{th:uniqueness}, condition (ii). Alternatively, we verify~\eqref{eq:semigroup_app} by explicit calculation:
\begin{align*}
\int_{\mathcal{Q}} db\, K(c|b)\,K(b|a) &= \int_{\mathcal{Q}} db \int_{\mathcal{A}} dA_1\, g(A_1;b|a)\, e^{i\eta A_1}\\
&\times \int_{\mathcal{A}} dA_2\, g(A_2;c|b)\, e^{i\eta A_2} \\
&= \int_{\mathcal{Q}} db \int_{\mathcal{A}} dA_1 \int_{\mathcal{A}} dA_2\, g(A_1;b|a)\\
& \times g(A_2;c|b)\, e^{i\eta (A_1 + A_2)}.
\end{align*}
Changing variables to $A = A_1 + A_2$ and using the composition law~\eqref{eq:composition_g}:
\[
\int_{\mathcal{Q}} db\, g(A_1;b|a)\, g(A - A_1;c|b) = g(A;c|a),
\]
we obtain
\[
\int_{\mathcal{Q}} db\, K(c|b)\,K(b|a) = \int_{\mathcal{A}} dA\, g(A;c|a)\, e^{i\eta A} = K(c|a).
\]
\end{proof}
This completes the derivation of the propagator as the unique dynamical object consistent with the axioms.
%----------------------------------------------------------------------------------------------
\section{Fourier Duality in Action Space}
\label{app:fourier_duality}

This appendix clarifies the purely mathematical role of the parameter $\eta$ as the variable conjugate to action under Fourier transformation. No physical assumption is introduced here; the discussion serves only to identify the structural meaning of $\eta_{\text{phase}}$ used in the definition of the propagator.

\subsection{Action--phase conjugacy}

Once the density of action states $g(A;b|a)$ is given and coherent combination of indistinguishable alternatives is required, the propagator takes the form
\begin{equation}
K(b|a;\eta)
=
\int_{\mathcal{A}} dA\, g(A;b|a)\, e^{i\eta A}.
\label{eq:fourier_prop}
\end{equation}
Mathematically, Eq.~\eqref{eq:fourier_prop} is a Fourier transform of $g(A;b|a)$ with respect to the action variable $A$.
The parameter $\eta$ therefore plays the role of the variable conjugate to action in the sense of harmonic analysis.

This identification is independent of the inferential origin of $\eta$. At this stage, $\eta$ is simply the phase parameter labeling the Fourier components of the action distribution. We denote it explicitly as $\eta_{\text{phase}}$ when its role as a Fourier conjugate is emphasized.

\subsection{Separation of roles}

It is important to stress that Fourier duality alone does \emph{not} fix the numerical value or physical scale of $\eta_{\text{phase}}$. The transform~\eqref{eq:fourier_prop} is well-defined for any real value of $\eta$, and Fourier analysis by itself does not select a preferred scale.

The physical meaning of $\eta$ arises only when this mathematical role is combined with:
\begin{enumerate}
    \item inference via maximum entropy in action space,
    \item the Cramér--Rao bound and operational indistinguishability of action,
    \item the requirement that interference occur precisely at the resolution limit of distinguishable alternatives.
\end{enumerate}

Fourier duality thus identifies \emph{what} $\eta$ is (a phase conjugate to action), but not \emph{how large} it must be.
The latter question is addressed in Appendix~\ref{app:scale_unification}, where all appearances of $\eta$ are shown to correspond to a single universal scale.

%---------------------------------------------------------------------------------------------
\section{Unification of Scales: Why the Same $\eta$ Appears Everywhere}
\label{app:scale_unification}

This appendix provides the detailed justification for why the parameter $\eta$ appearing in the MaxEnt distribution, the Cramér--Rao bound, and the phase character must be identified as a single universal constant.

\subsection{Three appearances of a scale parameter}

The parameter $\eta$ enters the theory in three distinct mathematical contexts:

\begin{enumerate}
    \item \textbf{MaxEnt distribution:} $P(A,b|a) \propto g(A;b|a)\, e^{-\eta_{\text{MaxEnt}} A}$

    \item \textbf{Cramér--Rao bound:} $\Delta A_{\min} \sim 1/\eta_{\text{CR}}$

    \item \textbf{Phase character:} $\chi(A) = e^{i\eta_{\text{phase}} A}$
\end{enumerate}

A priori, these could be three independent parameters. We now show that physical consistency requires their identification.

\subsection{Scale matching argument}

\noindent\textbf{Step 1: MaxEnt scale.}
The exponential $e^{-\eta_{\text{MaxEnt}} A}$ in the MaxEnt distribution suppresses large actions. The characteristic action scale is
\begin{equation}
A_{\text{typical}} \sim \frac{1}{\eta_{\text{MaxEnt}}}.
\label{eq:scale_maxent_app}
\end{equation}

\noindent\textbf{Step 2: Cramér--Rao indistinguishability scale.}
The Cramér--Rao bound $\Delta A \cdot \Delta \eta \geq 1$, combined with $\Delta \eta \sim \eta$ from Fourier conjugacy established in Appendix~\ref{app:fourier_duality}, gives
\begin{equation}
\Delta A_{\min} \sim \frac{1}{\eta_{\text{CR}}}.
\label{eq:scale_cr_app}
\end{equation}
This is the resolution scale: action differences smaller than $1/\eta_{\text{CR}}$ are operationally indistinguishable.

\noindent\textbf{Step 3: Phase accumulation scale.}
For interference effects to be observable, the phase difference $\Delta\phi = \eta_{\text{phase}}\Delta A$ must be of order unity when $\Delta A \sim \Delta A_{\min}$:
\begin{equation}
\eta_{\text{phase}} \cdot \Delta A_{\min} \sim 1 \quad\Rightarrow\quad
\eta_{\text{phase}} \sim \frac{1}{\Delta A_{\min}}.
\label{eq:scale_phase_app}
\end{equation}

\noindent\textbf{Step 4: Physical consistency.}
Comparing Eqs.~\eqref{eq:scale_maxent_app}--\eqref{eq:scale_phase_app}, all three parameters define the same physical scale:
\begin{equation}
\frac{1}{\eta_{\text{MaxEnt}}} \sim \frac{1}{\eta_{\text{CR}}} \sim \frac{1}{\eta_{\text{phase}}}
\sim \Delta A_{\min}.
\end{equation}

\begin{proposition}[Unification of scales]
\label{prop:eta_unification}
Physical consistency requires
\begin{equation}
\eta_{\text{MaxEnt}} = \eta_{\text{CR}} = \eta_{\text{phase}} \equiv \eta.
\label{eq:eta_unified}
\end{equation}
\end{proposition}

\begin{proof}
If the MaxEnt characteristic scale differed from the indistinguishability scale, the probability distribution would assign inconsistent weights to distinguishable versus indistinguishable action regions. If the phase accumulation scale differed
from the indistinguishability scale, interference would occur between operationally distinguishable alternatives, violating the meaning of indistinguishability. Therefore all three scales must coincide.
\end{proof}

The single parameter $\eta$ thus governs three aspects of the same physical phenomenon: the operational indistinguishability underlying quantum dynamics.

%----------------------------------------------------
\section{Operational Calibration of $\eta$: Double-Slit Example}
\label{app:double_slit}

This appendix demonstrates how $\eta$ can be determined experimentally from interference measurements, establishing the identification $\eta = 1/\hbar$.

\subsection{Experimental setup}

Consider the double-slit experiment with electrons:
\begin{itemize}
    \item Electron beam with momentum $p$
    \item Two slits separated by distance $d$
    \item Detection screen at distance $D$ from the slits ($D \gg d$)
    \item Observable: interference fringe spacing $\Delta y$
\end{itemize}

\subsection{Action difference between paths}

At screen position $y$, the path length difference is
\begin{equation}
\delta L \approx \frac{d \cdot y}{D} \quad (D \gg d, y).
\end{equation}
For a free particle with momentum $p$, the action difference is
\begin{equation}
\delta A = p \cdot \delta L = \frac{p \cdot d \cdot y}{D}.
\label{eq:action_diff_ds}
\end{equation}

\subsection{Phase difference and fringe spacing}

According to ITA, the phase difference between the two paths is
\begin{equation}
\Delta\phi = \eta \cdot \delta A = \frac{\eta \cdot p \cdot d \cdot y}{D}.
\end{equation}
Constructive interference occurs when $\Delta\phi = 2\pi n$, giving fringe positions
\begin{equation}
y_n = \frac{2\pi n D}{\eta \cdot p \cdot d}.
\end{equation}
The fringe spacing is therefore
\begin{equation}
\Delta y = \frac{2\pi D}{\eta \cdot p \cdot d}.
\label{eq:fringe_spacing_eta}
\end{equation}

\subsection{Identification with $\hbar$}

Standard quantum mechanics gives the fringe spacing as
\begin{equation}
\Delta y_{\text{QM}} = \frac{\lambda D}{d} = \frac{2\pi\hbar D}{p \cdot d},
\end{equation}
where $\lambda = h/p$ is the de Broglie wavelength. Comparing with Eq.~\eqref{eq:fringe_spacing_eta}:
\begin{equation}
\boxed{\eta = \frac{1}{\hbar}}
\end{equation}

\subsection{Operational procedure}

To determine $\eta$ without assuming $\hbar$:
\begin{enumerate}
    \item Perform double-slit experiment with electrons of known momentum $p$.
    \item Measure fringe spacing $\Delta y$.
    \item Compute $\eta = 2\pi D/(p \cdot d \cdot \Delta y)$.
    \item Verify universality by repeating with different $p$, $d$, $D$.
\end{enumerate}
The inferred value of $\eta$ is identified with the inverse of Planck’s constant,
\[
\eta = \frac{1}{\hbar}.
\]

\begin{remark}[De Broglie relation as a consequence]
The de Broglie relation $\lambda = h/p$ follows from ITA: interference maxima occur when $\eta \cdot p \cdot \delta L = 2\pi n$, so the wavelength satisfying $\eta \cdot p \cdot \lambda = 2\pi$ is $\lambda = 2\pi/(\eta p) = h/p$.
\end{remark}

%----------------------------------------------------
\section{Derivation of Short-Time $g$: Technical Details}
\label{app:shorttime_derivation}

This appendix provides technical details supporting the derivation in Section~\ref{sec:shorttime_g}.

\subsection{Gaussians from convolution stability}

\begin{theorem}[Gaussian uniqueness]
\label{th:gaussian_unique}
Here $T$ denotes a positive additive parameter (e.g.\ coarse-grained time or number of composition steps). Let $\{g_T\}_{T>0}$ be a family of probability densities on $\mathbb{R}$ forming a convolution semigroup with finite variance and continuity at the origin. Then $g_T$ is Gaussian with variance proportional to $T$.
\end{theorem}

\begin{proof}
We prove that Gaussians are the unique solution to the convolution stability requirement with finite variance.

\textbf{Setup:} Let $g_T$ be a family of probability densities on $\mathbb{R}$ satisfying:
\begin{enumerate}[label=(\alph*)]
    \item $g_{T_1} * g_{T_2} = g_{T_1+T_2}$ (convolution semigroup)
    \item $\int A^2 g_T(A)\, dA < \infty$ (finite variance)
    \item $g_T \to \delta_0$ weakly as $T \to 0$ (continuity at the origin).
\end{enumerate}

\textbf{Step 1: Characteristic function approach.}
Define $\hat{g}_T(k) = \int_{-\infty}^{\infty} e^{ikA} g_T(A)\, dA$. The convolution property becomes:
\[
\hat{g}_{T_1+T_2}(k) = \hat{g}_{T_1}(k) \cdot \hat{g}_{T_2}(k).
\]

\textbf{Step 2: Exponential form.}
Fix $k$ and define $\phi(T) = \hat{g}_T(k)$. The functional equation $\phi(T_1 + T_2) = \phi(T_1)\phi(T_2)$ with continuity implies:
\[
\phi(T) = e^{T\psi(k)}
\]
for some function $\psi(k)$, called the Lévy exponent.

\textbf{Step 3: Lévy--Khintchine representation.}
The Lévy--Khintchine theorem states that $\psi(k)$ must have the form:
\[
\psi(k) = i\mu k - \frac{\sigma^2 k^2}{2} + \int_{\mathbb{R}\setminus\{0\}} \left(e^{ikx} - 1 - ikx\mathbf{1}_{|x|<1}\right) \nu(dx),
\]
where $\mu \in \mathbb{R}$, $\sigma^2 \geq 0$, and $\nu$ is a Lévy measure satisfying $\int \min(1, x^2)\nu(dx) < \infty$.

\textbf{Step 4: Finite variance constraint.}
The variance of $g_T$ is:
\[
\text{Var}(g_T) = -\frac{d^2}{dk^2}\hat{g}_T(k)\Big|_{k=0} = T\left(\sigma^2 + \int x^2 \nu(dx)\right).
\]

For finite variance, we need $\int x^2 \nu(dx) < \infty$. But combined with the Lévy measure condition, this forces $\nu = 0$ (no jumps).

\textbf{Step 5: Drift constraint.}
Condition (iii) requires $g_T \to \delta_0$ weakly as $T\rightarrow 0^{+}$, meaning the distribution concentrates at the origin. For the Gussian semigroup with mean $\mu T$, the distribution is centered at $\mu T$ which vanisher as $T\rightarrow 0$ for any finite $\mu$; it is automatically satisfied. 

\textbf{Step 6: Conclusion.}
With $\nu = 0$ :
\[
\hat{g}_T(k) = e^{\imath \mu k T-\sigma^2 k^2 T/2},
\]
which is the characteristic function of a Gaussian with mean with mean $\mu T$ and variance $\sigma^2 T$:
\[
g_T(A) = \frac{1}{\sqrt{2\pi\sigma^2 T}} \exp\left(-\frac{(A-\mu T)^2}{2\sigma^2 T}\right).
\]
\end{proof}
The parameter $\mu$ represents the systematic rate of action accumulation (the drift in action space), which is related to the Lagrangian via $\mu=L(x,v)$ in the endpoint-dependent case (see Section~\ref{sec:shorttime_g}).

\subsection{Connection to the Cramér--Rao bound}

\begin{remark}[Consistency of Gaussian width with action resolution]
\label{prop:sigma_from_CR}
The width of the Gaussian action distribution must be consistent with the Cramér--Rao resolution scale.
\end{remark}

The Cramér--Rao bound (Corollary~\ref{cor:action_resolution}) establishes that action differences smaller than
\[
\Delta A_{\min} \sim \hbar
\]
are operationally indistinguishable.

For the Gaussian short-time distribution derived in the previous subsection, the characteristic spread of action values is
\[
\Delta A \sim \sigma\sqrt{T}.
\]

Self-consistency of the inferential description requires these two scales to match up to a numerical factor of order unity:
\[
\sigma\sqrt{T} \sim \Delta A_{\min} \sim \hbar.
\]

Indeed:
\begin{itemize}
    \item If $\sigma\sqrt{T} \ll \Delta A_{\min}$, the distribution is sharper than the operational resolution and carries no additional physical content.
    \item If $\sigma\sqrt{T} \gg \Delta A_{\min}$, the distribution spreads over many distinguishable action values, suppressing coherent behavior.
    \item The physically relevant regime is therefore $\sigma\sqrt{T} \sim \Delta A_{\min}$.
\end{itemize}

For infinitesimal durations $\Delta t$, this scaling implies
\[
\sigma^2 \Delta t \sim \frac{1}{\eta^2},
\]
consistent with the universal indistinguishability of infinitesimal action segments.

\textbf{Dimensional check:}
\[
[\sigma] = [A]/\sqrt{T},
\]
which is the correct dimension for an action-per-root-time parameter.

\subsection{Normalization from composition law}

\begin{proposition}[Normalization from composition consistency]
\label{prop:normalization}
Here we assume the quadratic short-time form dictated by the classical action and show that its normalization is uniquely fixed by the semigroup property. The prefactor $N(\Delta t) = (m/2\pi i\hbar\Delta t)^{1/2}$ is uniquely determined by requiring the composition law $K(c|a) = \int K(c|b)K(b|a)\,db$ to hold.
\end{proposition}

The prefactor $(m/2\pi i\hbar\Delta t)^{d/2}$ is determined by composition law consistency, not by quantum mechanical ``resolution of identity.''

\textbf{Derivation:}
The composition law $K(c|a) = \int db\, K(c|b)K(b|a)$ must hold for all times. For short-time propagators of the form:
\[
K(x_b, \Delta t | x_a) = N(\Delta t) \exp\left[\frac{im(x_b - x_a)^2}{2\hbar\Delta t}\right],
\]
the composition integral is a Gaussian convolution. Requiring $K(c, 2\Delta t | a) = \int db\, K(c, \Delta t | b)K(b, \Delta t | a)$ gives:
\begin{eqnarray*}
&& N(2\Delta t) \exp\left[\frac{im(c-a)^2}{4\hbar\Delta t}\right] = N(\Delta t)^2 \int db\\ &&\times \exp\left[\frac{im(c-b)^2 + im(b-a)^2}{2\hbar\Delta t}\right].
\end{eqnarray*}

The Gaussian integral evaluates to:
\[
\int db\, e^{\left[\frac{im(c-b)^2 + im(b-a)^2}{2\hbar\Delta t}\right]} = \sqrt{\frac{\pi i\hbar\Delta t}{m}} \exp\left[\frac{im(c-a)^2}{4\hbar\Delta t}\right].
\]

Matching both sides:
\[
N(2\Delta t) = N(\Delta t)^2 \cdot \sqrt{\frac{\pi i\hbar\Delta t}{m}}.
\]

This recursion, together with dimensional analysis and the classical limit, uniquely determines:
\[
N(\Delta t) = \left(\frac{m}{2\pi i\hbar\Delta t}\right)^{1/2}.
\]

The factor of $i$ in the denominator arises from the analytic continuation required for the oscillatory Gaussian integral, not from any quantum mechanical input.
% ---------------------------------------------------
\section{Diffusive Scaling of Action Variance}
\label{app:variance_scaling}

This appendix provides a rigorous derivation of the variance scaling $\sigma^2(T) \propto T$ used in Theorem~\ref{th:infinitesimal_indist} and throughout the text. This scaling is crucial for establishing universal infinitesimal indistinguishability.

\subsection{Statement of the result}

\begin{theorem}[Variance additivity for convolution semigroups]
\label{th:variance_additivity}
Let $g_T(A)$ be a family of probability densities on $\mathbb{R}$ indexed by time $T > 0$ satisfying:
\begin{enumerate}
    \item \textbf{Composition law (convolution semigroup):}
    \begin{equation}
    g_{T_1+T_2}(A) = \int_{-\infty}^{\infty} dA'\, g_{T_1}(A')\, g_{T_2}(A - A').
    \label{eq:convolution_semigroup_app}
    \end{equation}

    \item \textbf{Finite second moment:}
    \begin{equation}
    \int_{-\infty}^{\infty} A^2\, g_T(A)\, dA < \infty \quad \text{for all } T > 0.
    \label{eq:finite_second_moment_app}
    \end{equation}

    \item \textbf{Stationarity (time-translation invariance):} The family $\{g_T\}$ depends only on the duration $T$, not on absolute time.

    \item \textbf{Continuity at origin:} $g_T \to \delta(A)$ as $T \to 0$ in the weak sense.
\end{enumerate}

Then the variance satisfies:
\begin{equation}
\text{Var}[g_T] \equiv \int_{-\infty}^{\infty} A^2\, g_T(A)\, dA = \beta \cdot T
\label{eq:variance_linear_app}
\end{equation}
for some constant $\beta > 0$ with dimensions $[\text{Action}]^2/[\text{Time}]$. Importantly, no assumption of Gaussianity is made; the result follows solely from additivity, finite variance, and continuity.
\end{theorem}
This linear scaling underlies the $\sigma_{dt} \sim \sqrt{dt}$ behavior used in Theorem~\ref{th:infinitesimal_indist}, independently of the detailed form of $g_{T}$.

\subsection{Proof via cumulant additivity and central limit arguments}

\begin{proof}
We prove this in several steps.

\textbf{Step 1: Additivity of cumulants.}

For any convolution of probability densities $h = f * g$, the cumulants (connected moments) add:
\begin{equation}
\kappa_n[h] = \kappa_n[f] + \kappa_n[g].
\label{eq:cumulant_additivity}
\end{equation}

In particular, if $f$ and $g$ have zero mean, their variances (second cumulants) add:
\begin{equation}
\text{Var}[h] = \text{Var}[f] + \text{Var}[g].
\label{eq:variance_additivity_step}
\end{equation}

\textbf{Step 2: Applying to the composition law.}

From condition (i):
\[
g_{T_1+T_2} = g_{T_1} * g_{T_2}.
\]

By Step 1 (assuming zero mean, which can be enforced by centering):
\[
\text{Var}[g_{T_1+T_2}] = \text{Var}[g_{T_1}] + \text{Var}[g_{T_2}].
\]

Define $\sigma^2(T) := \text{Var}[g_T]$. This functional equation states:
\begin{equation}
\sigma^2(T_1 + T_2) = \sigma^2(T_1) + \sigma^2(T_2).
\label{eq:functional_variance}
\end{equation}

\textbf{Step 3: Solution of the functional equation.}

Equation~\eqref{eq:functional_variance} is Cauchy's functional equation for $\sigma^2(T)$. Under the assumption of continuity (which follows from condition (iv)), the only solutions are:
\begin{equation}
\sigma^2(T) = \beta T
\label{eq:cauchy_solution}
\end{equation}
for some constant $\beta \in \mathbb{R}$.

Since $\sigma^2(T)$ is a variance, we must have $\beta \geq 0$. The case $\beta = 0$ corresponds to $g_T = \delta(A - \mu(T))$ (deterministic), which contradicts the diffusive character implies by the convolution semigroup structure. Therefore $\beta > 0$.

\textbf{Step 4: Physical interpretation.}

The linear scaling $\sigma^2(T) = \beta T$ is characteristic of \emph{diffusive processes}. Physically:
\begin{itemize}
    \item Action accumulates as $A_T = \sum_{k=1}^{N} \delta A_k$ over $N$ time steps.
    \item Each increment $\delta A_k$ has variance $\sim \beta \,\delta t$.
    \item For weakly correlated (or independent) increments, the central limit theorem gives:
    \[
    \text{Var}(A_T) = \sum_{k=1}^{N} \text{Var}(\delta A_k) \approx N \cdot \beta\,\delta t = \beta T.
    \]
\end{itemize}

This is analogous to Brownian motion, where displacement variance grows linearly with time.
\end{proof}
Thus, Gaussian short-time action fluctuations are not an assumption but an unavoidable consequence of additive composition, finite variance, and continuity.

\subsection{Connection to the Gaussian form}

The diffusive scaling $\sigma^2 \propto T$, combined with the composition law, strongly constrains the functional form of $g_T$.

\begin{corollary}[Gaussian distribution from diffusive scaling]
\label{cor:gaussian_from_diffusion}
If $g_T$ satisfies the conditions of Theorem~\ref{th:variance_additivity} and has the diffusive scaling $\sigma^2(T) = \beta T$, then $g_T$ must be Gaussian:
\begin{equation}
g_T(A) = \frac{1}{\sqrt{2\pi\beta T}} \exp\left(-\frac{(A - \mu T)^2}{2\beta T}\right),
\label{eq:gaussian_from_diffusion}
\end{equation}
where $\mu$ is the drift rate (mean action per unit time).
\end{corollary}

\begin{proof}
This follows from the Lévy--Khintchine theorem (detailed in Appendix~\ref{app:shorttime_derivation}). A convolution semigroup with finite variance and linear variance growth is necessarily Gaussian. The jump term (Lévy measure $\nu$) must vanish due to finite variance, leaving only the Gaussian component.
\end{proof}

\subsection{Application to infinitesimal time steps}

For an infinitesimal time interval $dt$, the action increment $\delta A$ has:
\begin{align}
\text{Mean:} \quad &\langle \delta A \rangle = \bar{A}\,dt = L(x, v)\,dt, \\
\text{Variance:} \quad &\text{Var}(\delta A) = \beta\,dt.
\end{align}

Therefore, the standard deviation scales as:
\begin{equation}
\sigma_{dt} = \sqrt{\beta\,dt} \propto \sqrt{dt}.
\label{eq:sigma_sqrt_dt}
\end{equation}

This $\sqrt{dt}$ scaling is the key to universal infinitesimal indistinguishability (Theorem~\ref{th:infinitesimal_indist}): while the action difference scales as $\delta S \sim L\,dt \propto dt$, the resolution scale grows only as $\sqrt{dt}$, making their ratio vanish $(\delta S\equiv \delta A)$:
\begin{equation}
\frac{|\delta S_1 - \delta S_2|}{\sigma_{dt}} \sim \frac{O(dt)}{O(\sqrt{dt})} = O(\sqrt{dt}) \to 0.
\label{eq:ratio_vanishes}
\end{equation}

This scaling holds independently of the specific form of the Lagrangian, establishing infinitesimal indistinguishability as a universal property of action-based inference.

\subsection{Dimensional consistency}

The diffusion constant $\beta$ has dimensions:
\begin{equation}
[\beta] = \frac{[\text{Action}]^2}{[\text{Time}]} = \frac{(\text{J}\cdot\text{s})^2}{\text{s}} = \text{J}^2\cdot\text{s}.
\label{eq:beta_dimensions}
\end{equation}
 For consistency with the Cramér--Rao bound $\Delta A_{\min} \sim \hbar$, once the universal identification $\eta = 1/\hbar$ has been established, the diffusion constant must scale as
\begin{equation}
\beta \sim \frac{\hbar^2}{\tau_0},
\label{eq:beta_estimate}
\end{equation}
where $\tau_0$ is a characteristic microscopic time scale of the system. This relation fixes the order of magnitude of $\beta$ but does not introduce any new physical assumption.

In natural units where $\tau_0 \sim 1$, one has $\beta \sim \hbar^2$.

\subsection{Summary}

The diffusive scaling $\sigma^2(T) = \beta T$ is \emph{derived} from:
\begin{enumerate}
    \item The composition law (Axiom~\ref{ax:density}(b)),
    \item Finite variance (Axiom~\ref{ax:density}(e)),
    \item Additivity of variances for convolutions,
    \item Solution of Cauchy's functional equation.
\end{enumerate}

No additional assumptions about trajectories, quantum mechanics, or specific systems are needed. This is a universal property of convolution semigroups with finite second moments.

The $\sqrt{dt}$ growth of the width is the mathematical foundation of universal infinitesimal indistinguishability: while action differences scale as $O(dt)$, the resolution scale grows only as $O(\sqrt{dt})$. This mismatch forces coherent (complex amplitude) summation at the infinitesimal level, rather than allowing classical probability addition. This necessity is the essence of the path integral formulation.

%----------------------------------------------------
\section{Derivation of $L = T - V$ from Galilean Invariance}
\label{app:galilean_lagrangian}

This appendix shows that for Galilean-invariant systems, symmetry requirements uniquely determine the Lagrangian structure $L = T(v) - V(x)$ with kinetic energy $T = \tfrac{1}{2}mv^2$. The potential $V(x)$ remains system-specific input.

\subsection{Lagrangian structure from short-time limit}

For short time intervals, the mean action $\bar{A}(x_b, \Delta t | x_a)$ defines the effective Lagrangian governing the mean action via
\begin{equation}
\bar{A} = L(x, v)\,\Delta t + O(\Delta t^2),
\label{eq:L_emergence}
\end{equation}
where $v = (x_b - x_a)/\Delta t$ and $x = (x_a + x_b)/2$. At this stage, no restriction on the functional form of $L(x,v)$ is imposed.

\subsection{Galilean boost invariance}

Under a Galilean boost with velocity $u$:
\begin{equation}
x \to x' = x - ut, \qquad v \to v' = v - u.
\end{equation}
The action must transform by at most a total time derivative:
\begin{equation}
L(x, v) \to L'(x', v') = L(x, v) + \frac{d}{dt}f(x, t).
\end{equation}
For non-relativistic mechanics, this requires
\begin{equation}
L(x, v + u) = L(x, v) + mu \cdot v + \tfrac{1}{2}mu^2 + \text{const}
\label{eq:galilean_constraint}
\end{equation}
The constant term may depend on $u$, but not on $x, v$.
\begin{lemma}[Quadratic velocity dependence]
The constraint~\eqref{eq:galilean_constraint} implies $L$ is quadratic in $v$:
\begin{equation}
L(x, v) = a(x)\,v^2 + b(x) \cdot v + c(x),
\end{equation}
with $a(x) = \frac{1}{2}m$, where constancy follows from Galilean boost invariance and spatial homogeneity.
\end{lemma}

\begin{proof}
Expanding~\eqref{eq:galilean_constraint} to second order in $u$ and matching coefficients yields $\partial^2 L/\partial v^2 = m$ (constant) and determines the quadratic structure.
\end{proof}

\subsection{Time-reversal symmetry}

Under time reversal $t \to -t$, velocity changes sign: $v \to -v$. For time-reversal invariant Galilean systems without gauge structure (up to total derivatives), this requires $L(x,v)=L(x,-v)$, which forces $b(x)=0$.

\subsection{Final form}

Combining Galilean invariance and time-reversal (up to total time derivatives):
\begin{equation}
L(x, v) = \tfrac{1}{2}mv^2 + c(x).
\end{equation}
Defining $V(x) \equiv -c(x)$:
\begin{equation}
L(x, v) = \tfrac{1}{2}mv^2 - V(x) = T(v) - V(x)
\label{eq:L_TV_final}
\end{equation}

\begin{remark}[Universal vs.\ system-specific]
The kinetic term $T = \tfrac{1}{2}mv^2$ is \emph{derived} from Galilean symmetry (universal). The potential $V(x)$ is \emph{system input}: harmonic oscillator $V = \tfrac{1}{2}m\omega^2 x^2$, Coulomb $V = -ke^2/|x|$, etc. This parallels standard quantum mechanics, where one must specify $\hat{H}$.
\end{remark}

Thus, the classical Lagrangian emerges as the short-time, symmetry-constrained limit of action-based inference.

%----------------------------------------------------
\section{Derivation of Euler-Lagrange Equations from Composition}
\label{app:euler_lagrange}

This appendix provides the complete derivation of the Euler-Lagrange equations
from the propagator composition law via stationary phase analysis.

\subsection{Iterated composition and stationary phase}

Discretize the time interval $[0,T]$ into $N$ steps with $dt = T/N$ and positions
$x_0 = a, x_1, \ldots, x_N = b$. The propagator composition applied iteratively gives:
\begin{equation}
K(b,T|a) = \int dx_1 \cdots dx_{N-1} \prod_{k=0}^{N-1} K(x_{k+1}, dt | x_k).
\end{equation}
In the classical regime ($\eta\bar{A} \gg 1$):
\begin{equation}
e^{i\eta \bar{A}(b,T|a)} = \int dx_1 \cdots dx_{N-1}
\exp\!\left[i\eta \sum_{k=0}^{N-1} \bar{A}(x_{k+1}, dt | x_k)\right].
\end{equation}

For each intermediate point $x_j$ ($j = 1, \ldots, N-1$), stationary phase requires:
\begin{equation}
\frac{\partial}{\partial x_j}\left[\sum_{k=0}^{N-1} \bar{A}(x_{k+1}, dt | x_k)\right] = 0.
\end{equation}
Only two terms depend on $x_j$:
\begin{equation}
\frac{\partial \bar{A}(x_j, dt | x_{j-1})}{\partial x_j}
+ \frac{\partial \bar{A}(x_{j+1}, dt | x_j)}{\partial x_j} = 0.
\label{eq:two_terms_app}
\end{equation}

\subsection{Short-time expansion}

For small $dt$, the mean action for a single step is:
\begin{equation}
\bar{A}(x_{k+1}, dt | x_k) = L(\bar{x}_k, v_k)\, dt + O(dt^2),
\end{equation}
where $\bar{x}_k = (x_k + x_{k+1})/2$ is the midpoint and $v_k = (x_{k+1} - x_k)/dt$.

Computing the derivatives in~\eqref{eq:two_terms_app}:
\begin{align}
\frac{\partial \bar{A}(x_j, dt | x_{j-1})}{\partial x_j}
&= \frac{1}{2}\frac{\partial L}{\partial x}\bigg|_{\bar{x}_{j-1},v_{j-1}} dt
+ \frac{\partial L}{\partial v}\bigg|_{\bar{x}_{j-1},v_{j-1}}
+ O(dt), \\
\frac{\partial \bar{A}(x_{j+1}, dt | x_j)}{\partial x_j}
&= \frac{1}{2}\frac{\partial L}{\partial x}\bigg|_{\bar{x}_{j},v_{j}} dt
- \frac{\partial L}{\partial v}\bigg|_{\bar{x}_{j},v_{j}}
+ O(dt).
\end{align}

\subsection{Continuum limit}

Using the short-time derivatives obtained above (valid up to $O(dt)$), the stationary condition~\eqref{eq:two_terms_app} becomes
\begin{equation}
\left[\frac{1}{2}\frac{\partial L}{\partial x}\, dt + \frac{\partial L}{\partial v}\right]_{j-1}
+
\left[\frac{1}{2}\frac{\partial L}{\partial x}\, dt - \frac{\partial L}{\partial v}\right]_{j}
= O(dt^2),
\end{equation}
where the brackets $[\cdot]_j$ denote evaluation at the $j$-th time step (e.g.\ at the midpoint $(\bar x_j,v_j)$).

Rearranging and dividing by $dt$ gives
\begin{equation}
-\frac{1}{dt}\left[\frac{\partial L}{\partial v}\Big|_{j}
- \frac{\partial L}{\partial v}\Big|_{j-1}\right]
=
\frac{1}{2}\left[\frac{\partial L}{\partial x}\Big|_{j-1}
+ \frac{\partial L}{\partial x}\Big|_{j}\right]
+ O(dt).
\end{equation}

Taking the continuum limit $dt\to 0$ (equivalently $N\to\infty$), the left-hand side becomes $\frac{d}{dt}\left(\partial L/\partial v\right)$ and the right-hand side tends to
$\partial L/\partial x$, yielding
\begin{equation}
\frac{d}{dt}\frac{\partial L}{\partial \dot x}=\frac{\partial L}{\partial x},
\end{equation}
which is the Euler--Lagrange equation. \qed

\subsection{Classical action theorem}

Along the stationary trajectory $\{x_k^*\}$ selected by the stationary-phase condition, the mean action is additive:
\begin{equation}
\bar{A}(b,T|a)
=
\sum_{k=0}^{N-1} \bar{A}(x^*_{k+1}, dt | x^*_k)
=
\sum_{k=0}^{N-1} L(x^*_k, \dot{x}^*_k)\, dt + O(dt),
\end{equation}
where the $O(dt)$ term arises from the accumulation of subleading $O(dt^2)$ contributions.

Taking the continuum limit $N\to\infty$ yields
\begin{equation}
\bar{A}(b,T|a)\big|_{\text{stationary}}
=
\int_0^T L(x^*(t), \dot{x}^*(t))\, dt
\equiv S_{\text{cl}}(b,T|a).
\qed
\end{equation}

%----------------------------------------------------
\section{Hilbert Space Structure and Unitarity}
\label{app:hilbert_structure}

This appendix provides the formal definitions and complete proof of unitarity summarized in Section~\ref{sec:hilbert}.

\subsection{Function space and evolution operator}

We consider the space of complex-valued functions on $\mathcal{Q}$ for which the inner product~\eqref{eq:inner_product} is finite. The propagator $K(b|a)$ acts on functions $\psi: \mathcal{Q} \to \mathbb{C}$ via:
\begin{equation}
(U\psi)(b) = \int_{\mathcal{Q}} da\, K(b|a)\, \psi(a).
\label{eq:U_action}
\end{equation}
By construction, $U$ is a linear operator on this function space.  The inner product on this function space is:
\begin{equation}
\langle \phi | \psi \rangle = \int_{\mathcal{Q}} dq\, \phi^*(q)\, \psi(q).
\label{eq:inner_product}
\end{equation}

\subsection{Unitarity theorem and proof}

\begin{theorem}[Unitarity]
\label{th:unitarity}
If $g(A;b|a)$ satisfies the time-reversal symmetry~\eqref{eq:time_reversal} and the action functional is real under path reversal for real Lagrangians), then the evolution operator $U$ defined by~\eqref{eq:U_action} is unitary: $U^\dagger U = \mathbf{1}$.
\end{theorem}

\begin{proof}
We need to show that $\langle U(T)\phi | U(T)\psi \rangle = \langle \phi | \psi \rangle$ for all $\phi, \psi$ and all $T \geq 0$.

\textbf{Step 1: Relating $K^*$ to backward evolution.}
From time-reversal symmetry~\eqref{eq:time_reversal}:
\begin{align}
K^*(b,T|a) &= \left[\int dA\, g(A;b,T|a)\, e^{iA/\hbar}\right]^* \nonumber\\
&= \int dA\, g(A;b,T|a)\, e^{-iA/\hbar} \nonumber\\
&= \int dA\, g(A;a,T|b)\, e^{-iA/\hbar} \quad \text{[by~\eqref{eq:time_reversal}]}.
\label{eq:K_conjugate}
\end{align}

For systems with real Lagrangian $L(x, \dot{x}) = T(\dot{x}) - V(x)$, where the kinetic term $T$ is even in velocity, the action of a time-reversed path $\bar{\gamma}(t) \equiv \gamma(T - t)$ equals that of the original path: $S[\bar{\gamma}] = S[\gamma]$. This is because the kinetic energy $T(\dot{x})$ is invariant under $\dot{x} \to -\dot{x}$, so $L$ evaluated along the reversed path at each instant equals $L$ evaluated along the original path.

We now establish $K^{*}(b, T | a) = K(a, -T | b)$ directly from the short-time propagator structure derived in Section~VIII. For the short-time kernel~\ref{eq:K_shorttime_final}:
\begin{align}
  K(x_b, \Delta t \,|\, x_a) &= \left(\frac{m}{2\pi i\hbar\,\Delta t}\right)^{\!d/2}\nonumber\\&\times \exp\!\left[\frac{i}{\hbar} \left(\frac{m(x_b - x_a)^{2}}{2\,\Delta t}  V(\bar{x})\,\Delta t\right)\right],
\end{align}
where $\bar{x} = (x_a + x_b)/2$ is the midpoint.

Taking the complex conjugate:
\begin{align}
  K^{*}(x_b, \Delta t \,|\, x_a)
  &= \left(\frac{m}{2\pi(-i)\hbar\,\Delta t}\right)^{\!d/2}  \notag\\
  &\quad\times
    \exp\!\left[-\frac{i}{\hbar}
    \left(\frac{m(x_b - x_a)^{2}}{2\,\Delta t}
    - V(\bar{x})\,\Delta t\right)\right].
  \label{eq:Kstar}
\end{align}

Since $(x_b - x_a)^{2} = (x_a - x_b)^{2}$ and the midpoint $\bar{x} = (x_a + x_b)/2$ is symmetric under exchange, the classical short-time action satisfies
\begin{equation}
  S_{\mathrm{cl}}(x_b, \Delta t \,|\, x_a)
  = S_{\mathrm{cl}}(x_a, \Delta t \,|\, x_b)\,.
\end{equation}
Therefore the phase factor obeys
\begin{equation}
  \exp\!\bigl[-i\,S_{\mathrm{cl}}(x_b, \Delta t \,|\, x_a)/\hbar\bigr]
  = \exp\!\bigl[-i\,S_{\mathrm{cl}}(x_a, \Delta t \,|\, x_b)/\hbar\bigr].
\end{equation}

Comparing with the kernel for the reversed transition,
\begin{align}
  K(x_a, \Delta t \,|\, x_b)
  &= \left(\frac{m}{2\pi i\hbar\,\Delta t}\right)^{\!d/2}  \notag\\
  &\quad\times
    \exp\!\left[\frac{i}{\hbar}
    \left(\frac{m(x_a - x_b)^{2}}{2\,\Delta t}
    - V(\bar{x})\,\Delta t\right)\right],
\end{align}
we verify that
\begin{equation}
  K^{*}(x_b, \Delta t \,|\, x_a)
  = K(x_a, -\Delta t \,|\, x_b).
  \label{eq:shorttime_relation}
\end{equation}
The finite-time relation
\begin{equation}
  K^{*}(b, T \,|\, a) = K(a, -T \,|\, b)
\end{equation}
then follows by induction via the composition law. If the relation holds for time intervals $T_1$ and $T_2$, it holds for $T_1 + T_2$:
\begin{align}
  K^{*}(b,\, T_1 {+} T_2 \,|\, a)
  &= \biggl[\int_{\mathcal{Q}} \! dc \;
     K(b, T_2 \,|\, c)\, K(c, T_1 \,|\, a)\biggr]^{*}
     \notag\\[4pt]
  &= \int_{\mathcal{Q}} \! dc \;
     K^{*}(b, T_2 \,|\, c)\, K^{*}(c, T_1 \,|\, a)
     \notag\\[4pt]
  &= \int_{\mathcal{Q}} \! dc \;
     K(c, -T_2 \,|\, b)\, K(a, -T_1 \,|\, c)
     \notag\\[4pt]
  &= K(a,\, -(T_1 {+} T_2) \,|\, b)\,,
  \label{eq:induction}
\end{align}
where the second line uses linearity of integration and complex conjugation, the third line applies the induction hypothesis to each factor, and the fourth line invokes the composition law~ref{eq:semigroup} for the reversed transition.

This relation identifies $K(a,T|b)$ as the kernel of backward evolution over time $T$. %This key relation states that the adjoint of forward evolution from $a$ to $b$ equals forward evolution from $b$ to $a$.

\textbf{Step 2: Group extension and abstract argument.}
Time-reversal symmetry allows the extension of the forward semigroup $U(T)$, $T\ge 0$,
to a one-parameter group $U(t)$, $t\in\mathbb{R}$, with $U(-T)=U(T)^\dagger$.
This immediately implies
\[
U(T)^\dagger U(T) = U(-T)U(T) = U(0) = \mathbf{1}.
\]

For completeness, we now provide an explicit kernel-based verification of unitarity.
\textbf{Step 3: Computing $U^\dagger U$.}
The action of the adjoint operator is defined by the conjugate kernel with swapped indices.
\begin{align}
(U(T)^\dagger U(T) \psi)(a) &= \int_{\mathcal{Q}} db\, K^*(b, T \,|\, a)\nonumber\\\times & \left[ \int_{\mathcal{Q}} dc\, K(b, T \,|\, c)\, \psi(c) \right] \nonumber\\
&= \int_{\mathcal{Q}} dc\, \psi(c) \int_{\mathcal{Q}} db\, K^*(b, T \,|\, a)\nonumber\\&\times K(b, T \,|\, c) \nonumber\\
&= \int_{\mathcal{Q}} dc\, \psi(c) \int_{\mathcal{Q}} db\, K(a, -T \,|\, b)\, K(b, T \,|\, c)\nonumber \\& \quad \text{[using \eqref{eq:induction}]}\nonumber\\
&= \int_{\mathcal{Q}} dc\, \psi(c)\, K(a, 0 \,|\, c) \nonumber \\&\quad \text{[by semigroup property]}\nonumber\\
&= \int_{\mathcal{Q}} dc\, \psi(c)\, \delta(a-c) \nonumber\\
&= \psi(a).
\label{eq:unitarity_calc}
\end{align}

\textbf{Step 4: Physical interpretation and resolution.}
The composition $U(T)^\dagger U(T)$ corresponds to evolution forward by $T$ followed by evolution backward by $T$, which should return to the initial state. Mathematically, $K(a,T|b)$ with time-reversal represents backward evolution, so $U(T)^\dagger = U(-T)$ in the sense of reversed dynamics.This provides an intuitive interpretation of the abstract group-theoretic argument given in Step 2.

Equivalently, unitarity can be expressed as the completeness relation
\begin{equation}
\int db\, K^*(c,T|b)\, K(b,T|a) = \delta(c-a),
\label{eq:completeness}
\end{equation}
which follows from the orthogonality of phase factors $e^{iA/\hbar}$ combined with
time-reversal symmetry. Substituting into Step 2:
\begin{equation}
\langle U\phi | U\psi \rangle = \int da\, dc\, \phi^*(c)\psi(a)\, \delta(c-a) = \langle \phi | \psi \rangle.
\end{equation}
This proves $U^\dagger U = \mathbf{1}$.
\end{proof}

\subsection{Stone's theorem and the Hamiltonian}

As shown in the previous sections, reversibility and composition extend the propagator to a strongly continuous one-parameter unitary group on $L^2(\mathcal{Q})$.

\begin{theorem}[Existence of Hamiltonian]
\label{th:Stone}
If the propagator $K(b,t|a,0)$ defines a strongly continuous one-parameter unitary
group on $L^2(\mathcal{Q})$, then by Stone's theorem there exists a self-adjoint
operator $\hat{H}$ such that:
\begin{equation}
U(t) = e^{-i\hat{H}t/\hbar}.
\label{eq:Stone}
\end{equation}
\end{theorem}
The operator $\hat{H}$ is identified as the Hamiltonian governing time evolution.

\subsection{Technical remarks}

The above proof relies on the assumption that the Lagrangian is real and that the kinetic term $T(\dot{x})$ is even under velocity reversal. This condition is satisfied by all standard non-relativistic systems without velocity-dependent interactions, such as magnetic fields.

For systems with velocity-dependent potentials (e.g., charged particles coupled to a vector potential), the action of a time-reversed path differs from that of the original path by a boundary term associated with the gauge potential. As a consequence, the simple relation between the propagator and its complex conjugate no longer holds in its naive form.

Importantly, this does \emph{not} signal a breakdown of unitarity. Rather, it indicates that the proof must be formulated in a gauge-covariant manner, using the appropriately transformed propagator under time reversal. When this is done, unitary time evolution is fully preserved.

This situation is entirely analogous to the standard treatment of magnetic fields in the path-integral formulation of quantum mechanics, where gauge-covariant kernels ensure both gauge invariance and unitarity despite the presence of velocity-dependent interactions (e.g., as in the Feynman path-integral treatment of charged particles in electromagnetic fields).

%------------------------------------------------------------------------------------------------------
\section{Schrödinger Equation: Detailed Derivation}
\label{app:schrodinger_derivation}

Unitarity is not assumed but proven in Appendix~\ref{app:hilbert_structure} from time-reversal symmetry and the composition law.
The short-time kernel used in this derivation is \textbf{not assumed} from known quantum mechanics. It is the \emph{result} of the derivations in Section~\ref{sec:shorttime_g}:

\begin{enumerate}[label = \arabic{enumi}]
    \item \textbf{Gaussian form:} Derived from the Lévy--Khintchine theorem applied to the composition law (Theorem~\ref{th:gaussian_unique}).
    \item \textbf{Width $\sigma\sqrt{\Delta t} \sim \hbar$:} Derived from the Cramér--Rao bound (Proposition~\ref{prop:sigma_from_CR}).
    \item \textbf{Normalization factor $(m/2\pi i\hbar\Delta t)^{1/2}$:} Derived from composition law consistency (Proposition~\ref{prop:normalization}).
    \item \textbf{Kinetic term $\frac{m(x-x')^2}{2\Delta t}$:} Derived from Galilean invariance and time-reversal symmetry (Appendix~\ref{app:galilean_lagrangian}).
    \item \textbf{Potential term $V(x)$:} System specification (input), analogous to specifying the Hamiltonian in standard quantum mechanics.
\end{enumerate}

The kernel follows uniquely from the ITA axioms combined with composition, inference constraints, and spacetime symmetries.

\medskip
\subsection{Derivation of the Schrödinger equation from the derived kernel:}

Starting from the composition law:
\[
\psi(x,t+\Delta t) = \int dx'\, K(x,\Delta t|x')\, \psi(x',t)
\]
with the \emph{derived} short-time kernel:
\begin{eqnarray*}
K(x,\Delta t|x') &=& \left(\frac{m}{2\pi i\hbar\Delta t}\right)^{1/2} \exp\left[\frac{im(x-x')^2}{2\hbar\Delta t}\right.\\ &&\left. - \frac{i}{\hbar}V\left(\frac{x+x'}{2}\right)\Delta t\right].
\end{eqnarray*}

Substitute $x' = x - y$:
\begin{eqnarray*}
\psi(x,t+\Delta t) &=& \left(\frac{m}{2\pi i\hbar\Delta t}\right)^{1/2} \int dy\, e^{imy^2/(2\hbar\Delta t)}\\&& e^{-iV(x-y/2)\Delta t/\hbar}\, \psi(x-y,t).
\end{eqnarray*}

Expand for small $y$ and $\Delta t$:
\begin{align}
\psi(x-y,t) &\approx \psi - y\psi' + \frac{y^2}{2}\psi'', \\
V(x-y/2) &\approx V(x).
\end{align}
where higher-order terms in $y$ contribute only at $O(\Delta t^{3/2})$ and vanish in the continuum limit.
The key Gaussian integrals are:
\begin{align}
\left(\frac{m}{2\pi i\hbar\Delta t}\right)^{1/2} \int dy\, e^{imy^2/(2\hbar\Delta t)} &= 1, \label{eq:gauss0}\\
\left(\frac{m}{2\pi i\hbar\Delta t}\right)^{1/2} \int dy\, y^2 e^{imy^2/(2\hbar\Delta t)} &= \frac{i\hbar\Delta t}{m}. \label{eq:gauss2}
\end{align}

Collecting terms:
\[
\psi(x,t+\Delta t) = \psi(x,t) + \frac{i\hbar\Delta t}{2m}\psi''(x,t) - \frac{i\Delta t}{\hbar}V(x)\psi(x,t) + O(\Delta t^2).
\]

Rearranging and taking $\Delta t \to 0$:
\[
i\hbar\frac{\partial\psi}{\partial t} = -\frac{\hbar^2}{2m}\frac{\partial^2\psi}{\partial x^2} + V(x)\psi.
\]
At no stage is a quantum postulate or operator replacement rule assumed; the Schrödinger equation emerges solely as the continuum limit of the derived propagator.

\subsection{Hamiltonian via Legendre transform}

The operator $\hat{H} = -\frac{\hbar^2}{2m}\nabla^2 + V(x)$ appearing in the derived Schrödinger equation can be identified \emph{a posteriori} as the quantization of the classical Hamiltonian obtained via the Legendre transform of the emergent Lagrangian.

Starting from the emergent Lagrangian (Section~\ref{sec:lagrangian_emergence}):
\[
L(x,\dot{x}) = \frac{m}{2}\dot{x}^2 - V(x),
\]
the canonical momentum is:
\[
p = \frac{\partial L}{\partial\dot{x}} = m\dot{x}.
\]

The classical Hamiltonian is obtained by Legendre transform:
\[
H(x,p) = p\dot{x} - L = p \cdot \frac{p}{m} - \left(\frac{p^2}{2m} - V\right)
= \frac{p^2}{2m} + V(x).
\]
The Legendre transform plays no dynamical role in the derivation; it merely identifies the generator obtained from Stone's theorem with the classical Hamiltonian associated with the emergent Lagrangian. This identification is consistent with Stone's theorem (Appendix~\ref{app:hilbert_structure}), which guarantees that the generator of unitary time translations is self-adjoint. Promoting to operators with the standard position representation
$\hat{x}\psi(x) = x\psi(x)$ and momentum as the generator of translations
$\hat{p} = -i\hbar\nabla$:
\[
\hat{H} = \frac{\hat{p}^2}{2m} + V(\hat{x}) = -\frac{\hbar^2}{2m}\nabla^2 + V(x).
\]

This confirms that the operator in the derived Schrödinger equation is exactly the quantized Hamiltonian. 

\subsection{Canonical commutation relations}

The canonical commutation relation $[\hat{x}, \hat{p}] = i\hbar$ follows directly
from the position representation.

With $\hat{x}\psi(x) = x\psi(x)$ and $\hat{p}\psi(x) = -i\hbar\frac{d\psi}{dx}$,
we compute:
\begin{align}
[\hat{x}, \hat{p}]\psi &= \hat{x}\hat{p}\psi - \hat{p}\hat{x}\psi \nonumber\\
&= x\left(-i\hbar\frac{d\psi}{dx}\right) - \left(-i\hbar\right)\frac{d}{dx}(x\psi) \nonumber\\
&= -i\hbar x\frac{d\psi}{dx} + i\hbar\left(\psi + x\frac{d\psi}{dx}\right) \nonumber\\
&= i\hbar\psi.
\end{align}

Since this holds for all $\psi$, we have $[\hat{x}, \hat{p}] = i\hbar$.

This relation is not postulated but emerges from the representation of momentum as the generator of spatial translations, guaranteed by the strong continuity and unitarity of the translation group (Stone's theorem).

%----------------------------------------------------
\section{Worked Example: Free Particle}
\label{app:free_particle}

This appendix presents the complete ITA construction for the simplest nontrivial system: a free particle in one spatial dimension. We explicitly construct the density of action states $g(A;b|a)$, derive the propagator $K(b|a)$ via Fourier duality in action space, and verify full consistency with the standard quantum mechanical result.

\subsection{System specification}

For a free particle of mass $m$, the Lagrangian is
\begin{equation}
L(x,\dot{x}) = \frac{1}{2}m\dot{x}^2,
\label{eq:L_free}
\end{equation}
with no potential term ($V=0$). The dynamics is entirely kinetic.

\subsection{Construction of the action-state density}

According to the general short-time structure derived in Sec.~\ref{sec:shorttime_g}, the density of action states is Gaussian around the classical action.

\paragraph{Classical action.}
For a trajectory connecting $x_a$ to $x_b$ in a short time $\Delta t$, the velocity is $v=(x_b-x_a)/\Delta t$, and the classical action reads
\begin{equation}
S_{\mathrm{cl}}(x_b,\Delta t|x_a)
= L\,\Delta t
= \frac{m(x_b-x_a)^2}{2\Delta t}.
\label{eq:S_cl_free}
\end{equation}

\paragraph{Density of action states.}
The action-state density is taken as
\begin{equation}
g(A;x_b,\Delta t|x_a)
=
\frac{C(\Delta t)}{\sqrt{2\pi\sigma^2\Delta t}}
\exp\!\left[
-\frac{\left(A-S_{\mathrm{cl}}\right)^2}{2\sigma^2\Delta t}
\right],
\label{eq:g_free}
\end{equation}
where:
\begin{itemize}
\item $g$ is a \emph{density of action states}, not a probability density; 
\item $\sigma^2\Delta t$ is the variance of action increments; 
\item Cramér--Rao consistency implies $\sigma^2\Delta t \sim \hbar^2$ (up to a factor of order unity);
\item $C(\Delta t)$ is a normalization factor to be fixed by the composition law.
\end{itemize}

\subsection{Derivation of the propagator}

The propagator is defined as the Fourier transform of $g$ with respect to action:
\begin{equation}
K(x_b,\Delta t|x_a)
=
\int_{-\infty}^{\infty} dA\,
g(A;x_b,\Delta t|x_a)\,e^{iA/\hbar}.
\label{eq:K_from_g_free}
\end{equation}

Substituting Eq.~\eqref{eq:g_free}:
\begin{align}
K
&=
\frac{C(\Delta t)}{\sqrt{2\pi\sigma^2\Delta t}}
\int dA\,
\exp\!\left[
-\frac{(A-S_{\mathrm{cl}})^2}{2\sigma^2\Delta t}
+ \frac{iA}{\hbar}
\right].
\end{align}

Completing the square in the exponent,
\begin{align}
-\frac{(A-S_{\mathrm{cl}})^2}{2\sigma^2\Delta t}
+ \frac{iA}{\hbar}
&=
-\frac{1}{2\sigma^2\Delta t}
\left[
A-S_{\mathrm{cl}}-\frac{i\sigma^2\Delta t}{\hbar}
\right]^2\nonumber\\&
+ \frac{iS_{\mathrm{cl}}}{\hbar}
- \frac{\sigma^2\Delta t}{2\hbar^2}.
\end{align}

The Gaussian integral evaluates trivially, yielding
\begin{equation}
K(x_b,\Delta t|x_a)
=
C(\Delta t)\,
\exp\!\left[
\frac{iS_{\mathrm{cl}}}{\hbar}
- \frac{\sigma^2\Delta t}{2\hbar^2}
\right].
\label{eq:K_intermediate}
\end{equation}

\paragraph{Fixing the normalization.}
Since $\sigma^2\Delta t \sim \hbar^2$, the real exponential factor in Eq.~\eqref{eq:K_intermediate} is a short-time renormalization constant and can be absorbed into $C(\Delta t)$. The remaining time dependence of $C(\Delta t)$ is \emph{uniquely} fixed by the composition law (Proposition~\ref{prop:normalization}), yielding
\begin{equation}
C(\Delta t)
=
\left(\frac{m}{2\pi i\hbar\Delta t}\right)^{1/2}.
\end{equation}

Therefore, the free-particle propagator is
\begin{equation}
K(x_b,\Delta t|x_a)
=
\left(\frac{m}{2\pi i\hbar\Delta t}\right)^{1/2}
\exp\!\left[
\frac{im(x_b-x_a)^2}{2\hbar\Delta t}
\right].
\label{eq:K_free_final}
\end{equation}

\subsection{Verification and consistency}

Equation~\eqref{eq:K_free_final} coincides exactly with the standard free-particle propagator obtained in quantum mechanics via the path integral or by solving the Schrödinger equation.

\paragraph{Consistency checks.}
\begin{enumerate}
\item \textbf{Composition law:}
\[
\int dx_b\,
K(x_c,t_2|x_b)\,K(x_b,t_1|x_a)
=
K(x_c,t_1+t_2|x_a).
\]
Gaussian convolution reproduces the correct kernel with additive time. \checkmark

\item \textbf{Unitarity (kernel form):}
\[
\int dx_b\,
K^*(x_b,t|x_a)\,K(x_b,t|x_c)
=
\delta(x_a-x_c).
\]
Although $K(\cdot,t|x_a)$ is not square-integrable, it defines a unitary evolution
on $L^2$ wave functions. \checkmark

\item \textbf{Classical limit:}
As $\hbar\to 0$, the phase becomes rapidly oscillatory and stationary phase selects the classical straight-line trajectory $x_b-x_a=v\Delta t$. \checkmark
\end{enumerate}

This explicit construction illustrates, in the simplest setting, how the full quantum propagator emerges from action-based inference without postulating quantization rules.

%----------------------------------------------------
% Bibliography using BibTeX
\bibliography{ITA}

\end{document}